\documentclass[11pt,reqno,a4paper]{amsart}
\usepackage{amsmath,amssymb,amsthm,newlfont,graphicx,amscd,bbm,enumerate,galois,mathrsfs,xypic,geometry}
\usepackage{simplewick}
\usepackage{cite}
\usepackage{tcolorbox}

\newcommand{\lm}{\lambda}
\newcommand{\lmd}{\lambda}
\newcommand{\clm}{\Lambda}
\newcommand{\Lmd}{\Lambda}

\newcommand{\Omg}{\Omega}
\newcommand{\omg}{\omega}

\newcommand{\p}{\partial}
\newcommand{\al}{\alpha}
\newcommand{\afa}{\alpha}

\newcommand{\ve}{\varepsilon}
\newcommand{\veps}{\varepsilon}
\newcommand{\fai}{\varphi}
\newcommand{\rmi}{{\mathrm i}}
\newcommand{\rme}{{\mathrm e}}
\newcommand{\rmT}{{\mathrm T}}

\newcommand{\mcalA}{\mathcal{A}}

\newcommand{\mcalC}{\mathcal{C}}
\newcommand{\mcalD}{\mathcal{D}}
\newcommand{\mcalF}{\mathcal{F}}

\newcommand{\mcalH}{\mathcal{H}}
\newcommand{\mcalI}{\mathcal{I}}

\newcommand{\mcalK}{\mathcal{K}}
\newcommand{\mcalL}{\mathcal{L}}
\newcommand{\mcalP}{\mathcal{P}}
\newcommand{\mcalQ}{\mathcal{Q}}
\newcommand{\mcalR}{\mathcal{R}}
\newcommand{\mcalU}{\mathcal{U}}

\newcommand{\bbC}{\mathbb{C}}
\newcommand{\bbZ}{\mathbb{Z}}
\newcommand{\bft}{\mathbf{t}}

\newcommand{\bsH}{\boldsymbol{H}}
\newcommand{\bsP}{\boldsymbol{P}}
\newcommand{\bsQ}{\boldsymbol{Q}}

\newcommand{\sfp}{\mathsf{p}}
\newcommand{\sfq}{\mathsf{q}}
\newcommand{\sfr}{\mathsf{r}}

\newcommand*{\pp}[1]
  {\frac{\partial   }
        {\partial #1}
  }

\newcommand*{\pfrac}[2]
  {\frac{\partial #1}
        {\partial #2}
  }

\newcommand*{\Bigset}[2]
  {
   \left\{ #1 \,\middle|\, #2 \right\}
  }

\newcommand{\beq}{\begin{equation}}
\newcommand{\eeq}{\end{equation}}

\DeclareMathOperator{\ad}{ad}
\DeclareMathOperator{\res}{Res}
\DeclareMathOperator{\diag}{diag}

\DeclareMathOperator{\Der}{Der}
\DeclareMathOperator{\nd}{d\!}
\DeclareMathOperator{\td}{d\!}

\newtheorem{thm}{Theorem}[section]

\newtheorem{rmk}[thm]{Remark}
\newtheorem{cor}[thm]{Corollary}
\newtheorem{lem}[thm]{Lemma}

\newtheorem{prop}[thm]{Proposition}

\newtheorem{ex}{Example}[section]

\numberwithin{equation}{section}

\allowdisplaybreaks[4]

\begin{document}
\title[Solutions of the Loop Equations of Generalized Frobenius Manifolds]{Solutions of the Loop Equations of a Class of Generalized Frobenius Manifolds}
\author{Si-Qi Liu}
\address{Department of Mathematical Sciences,
Tsinghua University, Beijing 100084, P.R. China}
\email{liusq@tsinghua.edu.cn}
\author{Haonan Qu}
\address{School of Mathematical Sciences,
Peking University, Beijing 100871, P.R. China}
\email{qhn1121@pku.edu.cn}
\author{Yuewei Wang}
\address{Department of Mathematical Sciences,
Tsinghua University, Beijing 100084, P.R. China}
\email{wlw18@mails.tsinghua.edu.cn}
\author{Youjin Zhang}
\address{Department of Mathematical Sciences,
Tsinghua University, Beijing 100084, P.R. China}
\email{youjin@tsinghua.edu.cn}
\keywords{Generalized Frobenius manifold, Loop equation, the Volterra hierarchy, the $q$-deformed KdV hierarchy, Principal Hierarchy}
\maketitle

\begin{abstract}
We prove the existence and uniqueness of solution of the loop equation associated with a semisimple generalized Frobenius manifold with non-flat unity, and show, for a particular example of one dimensional generalized Frobenius manifold, that the deformation of the Principal Hierarchy induced by the solution of the loop equation is the extended $q$-deformed KdV hierarchy.
\end{abstract}

\tableofcontents

\section{Introduction}
Since Dubrovin introduced the notion of Frobenius manifold in the early 1990s \cite{Du1, Du2, Dubrovin2DTFT}, the relationship between Frobenius manifolds and hierarchies of integrable evolutionary PDEs has been an important subject in the research field of mathematical physics, see for example \cite{Lorenzoni3, Lorenzoni2, Bakalov, Brini1, Brini2, Buryak2012, Buryak2015, Buryak2019, Buryak2021, carlet2018deformations, DZ1999, normalform, DZ2004, DLZ06, DLZ08, DLZ2018, FJRW, Getzler, Getzler04, GiventalMilanov, liu2013bihamiltonian, GFM} and references therein. As it was shown by Dubrovin \cite{Du1, Du2, Dubrovin2DTFT}, one can associate with any given Frobenius manifold a bihamiltonian integrable hierarchy of hydrodynamic type, which is called the Principal Hierarchy of the Frobenius manifold. Under the assumption of semisimplicity of the Frobenius manifold, it is shown in \cite{normalform} that the Principal Hierarchy has a so-called topological deformation which is a bihamiltonian integrable hierarchy of KdV type \cite{LWZ1}. If the semisimple Frobenius manifold arises from the quantum cohomology of a smooth projective variety $X$, then we know from the construction of the topological deformation of the Principal Hierarchy that the logarithm of a particular tau function of the integrable Hierarchy yields the generating function of the Gromov-Witten invariants of $X$. In particular, when $X$ is a point we have the one-dimensional Frobenius manifold, and the bihamiltonian integrable hierarchy is the KdV hierarchy \cite{Kon, Witten}; when $X$ is the complex projective line, the Frobenius manifold is two-dimensional, and the associated
bihamiltonian integrable hierarchy is the extended Toda hierarchy \cite{exttoda, DZ2004, Getzler, OP, Z2002}.
For a semisimple Frobenius manifold, the Principal Hierarchy is related with its topological deformation via a quasi-Miura transformation
which is given by the unique solution of its loop equation. As it is introduced in \cite{normalform}, the loop equation of a Frobenius manifold encodes the condition that the actions of the Virasoro symmetries of the Principal Hierarchy on the tau function of its topological deformation are linear, i.e., they are given by the actions of certain second order linear operators, which are called the Virasoro operators, on the tau function.

An analogue of the relationship between semisimple Frobenius manifolds and bihamiltonian integrable hierarchies is proposed in \cite{GFM} for a class of generalized Frobenius manifolds, which satisfy all the axioms of Dubrovin's definition of Frobenius manifolds but the flatness of the unit vector fields \cite{Dubrovin2DTFT}, i.e., the unit vector fields are non-flat with respect to the flat metrics of the manifolds. For such a generalized Frobenius manifold $M$, one can
construct a bihamiltonian integrable hierarchy of hydrodynamic type via the deformed flat coordinates of $M$, just as one does for the construction of the Principal Hierarchy of a usual Frobenius manifold. However, due to the non-flatness of the unit vector field, this bihamiltonian integrable hierarchy does not contain the flow that is given by the translation along the spatial variable. By applying the bihamiltonian recursion relation to this special flow one obtains an infinite number of flows which are symmetries of the integrable hierarchy. These flows play an important rule in our understanding of properties of the Virasoro symmetries of the bihamiltonian integrable hierarchy of hydrodynamic type, and they are added to this integrable hierarchy to form the Principal Hierarchy of the generalized Frobenius manifold.

The Principal Hierarchy of a generalized Frobenius manifold also possesses a tau structure, and its Virasoro symmetries acts nonlinearly on the tau function. As for a usual Frobenius manifold, we construct the topological deformation of the Principal Hierarchy by requiring that the Virasoro symmetries acts linearly on the tau function of the deformed integrable hierarchy. Namely, if we denote by $\tau^{[0]}$ and $\tau$ the tau functions of the Principal Hierarchy and its topological deformation, then we require that these tau functions are related by the following genus expansion formula:
\[\log\tau=\ve^{-2}\log\tau^{[0]}+\sum_{k\ge 1}\ve^{k-2}\mcalF^{[k]}(v;v_x,\dots,v^{(N_k)}),\]
and the infinitesimal actions of the Virasoro symmetries of the Principal Hierarchy on the tau function of the deformed integrable hierarchy are given by
\[\tau\to\tau+\delta L_m\tau,\quad m=-1, 0, 1,\dots.\]
Here $v=(v^1,\dots, v^n)$ is a solution of the Principal Hierarchy that depends on the spatial variable $x$ and time variables $t^{\al,p}$ and $L_m$, called the Virasoro operators, are second order linear differential operators of the time variables. This condition of linearization of the Virasoro symmetries leads to a system of linear equations for the gradients of the functions $\mcalF^{[k]}$ for $k\ge 1$, and we call
it the loop equation of the generalized Frobenius manifold. See Theorem 10.5 of \cite{GFM} for the explicit form of the loop equation.

In the present paper, we prove the existence and uniqueness of solution of the loop equation of a semisimple generalized Frobenius manifold, and prove the Conjecture 11.1 of \cite{GFM}. This conjecture asserts that the topological deformation of the Principal Hierarchy of a certain one-dimensional generalized Frobenius manifold yields the Volterra hierarchy and the $q$-deformed KdV hierarchy.

We organize the paper as follows: In Sect.\,2, we present some important properties of a semisimple generalized Frobenius manifold. In Sect.\,3, we prove the existence and uniqueness of solution of the loop equation, the main results are given in Theorem \ref{thm:the uniqueness} and Theorem \ref{thm-existence}. In Sect.\,4, we prove the above-mentioned conjecture on the topological deformation of the Principal Hierarchy of a one-dimensional generalized Frobenius manifold. In Sect.\,5, we give some concluding remarks.

\section{Semisimple generalized Frobenius manifolds}
In this section, we present some basic properties of a semisimple generalized Frobenius manifold with non-flat unity, which will be used in Sect.\,3 to prove the uniqueness and existence of solution of the associated loop equation. Most of these properties are shared by a usual semisimple Frobenius manifold, we refer readers to \cite{Dubrovin2DTFT, Dubrovinpainleve} for details of their derivations.

Let $M$ be an $n$-dimensional generalized Frobenius manifold with non-flat unity \cite{Dubrovin2DTFT, GFM}. We will simply call it a generalized Frobenius manifold or GFM for short in what follows. Then we have an operation of multiplication $``\cdot"$ on the tangent spaces of $M$ which yields, together the unit vector field $e$ and the flat metric $\eta$ of $M$, a Frobenius algebra structure on each tangent space of $M$. Denote by $\nabla$ the Levi-Civita connection of $\eta$, then the non-flatness of the unit vector field $e$ means that $\nabla e\ne 0$. We also have the Euler vector field $E$ which is associated with the quasi-homogeneity properties of the flat metric $\eta$ and the operation of multiplication on the tangent spaces of $M$, and satisfies the condition $\nabla^2 E=0$. As for a usual Frobenius manifold, we have the deformed flat connection $\tilde\nabla$ on $M$ which is defined by
\[\tilde\nabla_a b=\nabla_a b+z a\cdot b,\quad a, b\in\textrm{Vect}(M).\]
It can be extended to a flat connection of $M\times \mathbb{C}^*$ by defining
\begin{align*}
\tilde{\nabla}_a\frac{\nd}{\nd z}=0,\quad \tilde{\nabla}_{\frac{\nd}{\nd z}}\frac{\nd}{\nd z}=0,\quad \tilde{\nabla}_{\frac{\nd}{\nd z}}b=\partial_z b+E\cdot b-\frac{1}{z}\mu b,
\end{align*}
here $a, b$ are vector fields on $M\times \mathbb{C}^*$ with zero  components along $\mathbb{C}^*$, $z$ is the coordinate of $\mathbb{C}^*$, and $\mu$ is the grading operator
\begin{equation}\label{mu}
    \mu=\frac{2-d}{2}-\nabla E
\end{equation}
 with $d$ being the charge of $M$.

\subsection{The canonical coordinates of a semisimple GFM}
Following the definition of the semisimplicity of a usual Frobenius manifold \cite{Dubrovin2DTFT, Dubrovinpainleve}, we call the generalized Frobenius manifold $M$ semisimple if the Frobenius algebras defined on $T_p M$ are semisimple for generic $p\in M$. As it is shown in \cite{Dubrovin2DTFT, Dubrovinpainleve}, one can find a system of local coordinates $u_1,\dots, u_n$ in a certain neighborhood of a generic point of $M$ such that
\begin{equation}
  \pp{u_i}\cdot\pp{u_j}=\delta_{ij}\pp{u_i},\quad i, j=1,\dots,n.
\end{equation}
These coordinates are called canonical coordinates of $M$, in which the unit vector field has the form
\[e=\sum_{i=1}^n \frac{\p}{\p u_i}\]
and the flat metric $\eta$ is diagonal, i.e.,
\begin{equation}\label{def of h_i}
\left\langle\frac{\p}{\p u_i},\frac{\p}{\p u_j}\right\rangle=\eta\left(\frac{\p}{\p u_i},\frac{\p}{\p u_j}\right)=h_i\delta_{ij},\quad i, j=1,\dots,n.
\end{equation}
It is shown in \cite{GFM} that the unit vector field $e$ is the gradient of a function $\theta_{0,0}(u)$ with respect to the metric $\eta$, so the functions $h_i$ can be represented in the form
\begin{equation}\label{zh-1-1}
  h_i=\pfrac{\theta_{0,0}}{u_i},\quad i=1,2,...,n.
\end{equation}
From \cite{Dubrovin2DTFT, Dubrovinpainleve} we know that one can choose the canonical coordinates of $M$ in such a way that in these coordinates the Euler vector field has the form
\[E=\sum_{i=1}^{n}u_i\pp{u_i}.\]
In another word, we can take the eigenvalues of the
 operator
of multiplication by the Euler vector field $E$ as canonical coordinates of $M$.
We will fix such a choice of canonical coordinates of $M$ which is assumed to be semisimple  in what follows.

Let us also fix a system of local flat coordinates $v^1,\dots, v^n$ of $M$ with respect to its flat metric $\eta$, and denote
by $f_1,...,f_n$ the orthonormal frame
\begin{equation}\label{230820-2159-1}
  f_i=\frac{1}{\sqrt{h_i}}\pp{u_i},\quad i=1,\dots,n.
\end{equation}
Then we have
\begin{equation}\label{def of psi i afa}
  \left(\pp{v^1},\dots,\pp{v^n}\right)=(f_1,\dots,f_n)\Psi,
\end{equation}
where $\Psi=(\psi_{i\afa})$ is an $n\times n$ matrix valued function
which satisfies the relation
\begin{equation}\label{ortho of Psi}
  \Psi^{\rmT}\Psi=\eta.
\end{equation}
Here and in what follows we also use $\eta$ to denote the constant matrix $(\eta_{\al\beta})$ with
\[\eta_{\al\beta}=\eta\left(\frac{\p}{\p v^\al},\frac{\p}{\p v^\beta}\right).\]
In terms of the matrix $\Psi$, the elements of the Jacobi matrix of the change of coordinates between the canonical and the flat ones are given by
\begin{equation}\label{zh-1-4}
\frac{\p u_i}{\p v^\al}=\frac{\psi_{i\al}}{\sqrt{h_i}},\quad \frac{\p v^\al}{\p u_i}=\sqrt{h_i}\psi_i^\al,\quad i,\al=1,\dots, n,
\end{equation}
where $\psi^\al_i=\eta^{\al\gamma}\psi_{i\gamma}$, $(\eta^{\al\beta})=(\eta_{\al\beta})^{-1}$, and we assume summation with respect to repeated upper and lower Greek indices henceforth. The structure constants of the multiplication
\[\frac{\p}{\p v^\al}\cdot\frac{\p}{\p v^\beta}=c^\gamma_{\al\beta}\frac{\p}{\p v^\gamma},\quad \al, \beta=1,\dots,n\]
can be represented in the form
\begin{equation}
  c_{\afa\beta}^\gamma
=
  \sum_{i=1}^{n}
    \frac{\psi_{i\afa}\psi_{i\beta}\psi_i^\gamma}{\sqrt{h_i}}.
\end{equation}

From the relation \eqref{def of psi i afa} we know that the functions $\psi_{i\al}$ satisfy the equations
\begin{equation}\label{p_i psi}
  \p_i\psi_{j\afa}=\gamma_{ij}\psi_{i\al},\quad
  \p_i\psi_{i\al}= -\sum_{k\neq i}\gamma_{ki}\psi_{k\afa},\quad  i\ne j,
\end{equation}
here
\begin{equation}\label{def of rotation coef}
   \gamma_{ij}=\frac{1}{\sqrt{h_j}}\pfrac{\sqrt{h_i}}{u_j}, \quad i\neq j
\end{equation}
are the rotation coefficients of the diagonal metric $\eta=\sum_{i=1}^{n}h_i\td u_i^2$, which has the symmetry property
\begin{equation}\label{zh-1-5}
\gamma_{ij}=\gamma_{ji},\quad i\ne j
\end{equation}
due to \eqref{zh-1-1}, and satisfy the Darboux-Egoroff system
\begin{align}
&\frac{\p\gamma_{ij}}{\p u_k}=\gamma_{ik}\gamma_{kj},\quad \textrm{for distinct}\ i, j, k=1,\dots,n,\\
&\frac{\p\gamma_{ij}}{\p u_i}+\frac{\p\gamma_{ij}}{\p u_j}=-\sum_{l\ne i, j} \gamma_{il}\gamma_{lj},\quad \textrm{for distinct}\ i, j=1,\dots,n.
\end{align}
From \eqref{p_i psi} it also follows that
\[\p_e\Psi=0,\]
here $e$ is the unit vector field. If we represent $e$ in terms of the flat coordinates as follows:
\[e=e^\afa\pp{v^{\afa}},\]
then from the relations
\begin{align*}
  f_i&=\,e\cdot f_i=e^\afa\pp{v^{\afa}}\cdot f_i
  =
    e^\afa\sum_{j=1}^{n}\psi_{j\afa}f_j\cdot f_i
  =
    \frac{1}{\sqrt{h_i}}e^\afa\psi_{i\afa}f_i
\end{align*}
we arrive at
\begin{equation}
  \sqrt{h_i}=e^\afa\psi_{i\afa},\quad i=1,\dots,n.
\end{equation}

\subsection{The isomonodromic tau-function of a GFM}
Recall that the horizontal section $\xi=(\xi^1,...,\xi^n)^{\rmT}$
of the deformed flat connection $\tilde\nabla$ of the generalized Frobenius manifold $M$ satisfies the linear system
 \cite{Dubrovin2DTFT, Dubrovinpainleve,normalform,GFM}
\begin{align}\label{horizontal section equation 1}
     \p_{\afa}\xi&=z\mcalC_\afa\xi,  \\
    \p_z\xi&=\left(\mcalU+\frac{\mu}{z}\right)\xi,\label{horizontal section equation 2}
\end{align}
here the matrices $\mcalC_\afa$ and $\mcalU$ are representations, in the frame given by the flat coordinates, of the operators of mutiplication by the vector field $\frac{\p}{\p v^\al}$
and the Euler vector field $E$ respectively.
After the gauge transformation $y=\Psi\xi$, the system \eqref{horizontal section equation 1}, \eqref{horizontal section equation 2} reads
\begin{align}
     \pp{u_i}y&=\, (zE_i+V_i)y,\\
    \pp{z}y&=\, \left(U+\frac{V}{z}\right)y,\label{zh-1-2}
\end{align}
where the matrices $E_i$ are given by $(E_i)_{kl}=\delta_{ik}\delta_{il}$, and
\[ U=\,\Psi\mcalU\Psi^{-1}=\diag(u_1,\dots,u_n),\quad V=\Psi\mu\Psi^{-1},\quad V_i=\ad_U^{-1}[E_i,V].\]
The matrix $V$ is anti-symmetric and it is related with the rotation coefficients by
\begin{equation}\label{zh-1-6}
V=[\Gamma, U],
\end{equation}
here $\Gamma=(\gamma_{ij}) $. It satisfies the isomonodromy deformation equations
\begin{equation}\label{zh-1-3}
\frac{\p V}{\p u_i}=[V_i, V],\quad i=1,\dots,n
\end{equation}
of the linear system \eqref{zh-1-2}. These equations can be represented as the following Hamiltonian systems:
\[\frac{\p V}{\p u_i}=\{V, H_i\},\quad i=1,\dots,n\]
with the linear Poisson bracket on $so(n)$ defined by
\[\{V_{ij},V_{kl}\}=V_{il}\delta_{jk}-V_{jl}\delta_{ik}+V_{jk}\delta_{il}-V_{ik}\delta_{jl},\]
and the Hamiltonians given by
\begin{equation}
  H_i=\frac12\sum_{j\neq i}\frac{V_{ij}^2}{u_i-u_j},\quad i=1,\dots,n.
\end{equation}
The tau function $\tau_I=\tau_I(u)$ of these Hamiltonian systems defined by
\begin{equation} \label{isomonodromic tau function}
  \td\log\tau_I=\sum_{i=1}^{n}H_i\td u_i
\end{equation}
is called the isomonodromic tau function of the semisimple generalized Frobenius manifold $M$.

\subsection{Periods of a semisimple GFM}
Let $g$ be the intersection form of the generalized Frobenius manifold $M$. In the flat coordinates it is represented by the matrix
$g=(g^{\al\beta})$, where
\[g^{\al\beta}=\imath_E\left(\nd v^\al\cdot\nd v^\beta\right)=\eta^{\al\gamma}\mcalU^\beta_\gamma.\]
We denote
\begin{equation}
  \Sigma_\lmd=
  \Bigset{v\in M}{\det (g^{\afa\beta}-\lmd\eta^{\afa\beta})|_v=0},\quad \lmd\in\bbC,
\end{equation}
then the inverse matrix
\[
  (g_{\afa\beta}(\lmd))=(g^{\afa\beta}-\lmd\eta^{\afa\beta})^{-1}
\]
defines a flat metric $(\ ,\, )_\lm$ on $M\setminus\Sigma_\lmd$. In the flat coordinates $v^1,\dots, v^n$ of $M$, the contravariant components of the Levi-Civita connection $\nabla(\lm)$ of this metric are given by
\[\nabla(\lm)^\al \nd v^\beta=(g^{\al\gamma}-\lm \eta^{\al\gamma})\nabla(\lm)_\gamma\nd v^\beta=\Gamma^{\al\beta}_\gamma \nd v^\gamma,\]
where $\Gamma^{\al\beta}_\gamma$ are the contravariant coefficients of the Levi-Civita connection of $g$.
There exists a system of flat local coordinates
$p_1(v;\lm),\dots, p_n(v;\lm)$ of $(\ ,\,)_\lm$ on an open subset of $M\setminus \Sigma_\lm$,
called the \textit{periods} of the generalized Frobenius manifold $M$,  which satisfy the equations
\begin{equation}\label{flat 1-forms}
\nabla(\lm)\nd p_\al=0,\quad \al=1,\dots,n.
\end{equation}

Let $p$ be a period of the semisimple generalized Frobenius manifold $M$, then its gradient with respect to the flat metric $\eta$ can be represented in the orthonormal frame as follows:
\begin{equation}
  \nabla p = \sum_{i=1}^{n}\phi_i f_i,
\end{equation}
then we have
\begin{equation}\label{230821-1052}
  \p_\afa p=\sum_{i=1}^{n}\psi_{i\afa}\phi_i,\qquad
  \p_i p=\sqrt{h_i}\phi_i,
\end{equation}
where we use the short notation $\p_\afa=\pp{v^\afa}$ and $\p_i=\pp{u_i}$.

\begin{prop}
The period $p$ of $M$ satisfies the following system of Gauss-Manin equations:
\begin{align}
  \pfrac{\phi_j}{u_i}&=\gamma_{ij}\phi_i,\quad i\neq j, \label{ss. GM-1}\\
  \pfrac{\phi_i}{u_i}&=\,
    \frac{1}{\lmd-u_i}
      \left(
        \frac 12\phi_i+\sum_{j=1}^{n}V_{ij}\phi_j
      \right)
     -\sum_{j\neq i}\gamma_{ij}\phi_j,  \label{ss. GM-2}\\
  \pfrac{\phi_i}{\lmd}&=\,
    -\frac{1}{\lmd-u_i}
     \left(
       \frac12\phi_i+\sum_{j=1}^{n}V_{ij}\phi_j
     \right), \label{ss. GM-3}
\end{align}
where $V_{ij}$ is the $(i,j)$-th element of the matrix $V$.
\end{prop}

\begin{proof}
  Recall that the gradient $\nabla p=\sum_{i=1}^{n}\phi_if_i$ of the period $p$
satisfies the following Gauss-Manin equations \cite{Dubrovinpainleve, GFM}
  \begin{align}
    &(\mcalU-\lmd)\nabla_i
      \left(
        \sum_{j=1}^{n}\phi_jf_j
      \right)
    +\mcalC_i\left(\mu+\frac12\right)
    \left(
      \sum_{j=1}^{n}\phi_jf_j
    \right)=0, \\
  &(\mcalU-\lmd)\pp\lmd\left(
    \sum_{j=1}^{n}\phi_jf_j\right)
  =
    \left(\mu+\frac12\right)
    \left(
      \sum_{j=1}^{n}\phi_jf_j
    \right),
  \end{align}
where the operators $\mcalU$, $\mcalC_i$ and $\mu$ are defined by
\[\mcalU f_j= E\cdot f_j,\quad  \mcalC_if_j=\pp{u_i}\cdot f_j=\delta_{ij}f_j,\quad \mu f_i=\sum_{i=1}^{n}V_{ij}f_i.
\]
Then this proposition is obtained from \eqref{p_i psi}
by a straightforward calculation.
\end{proof}

Now let us fix a basis of periods $p_1,\dots, p_n$, and denote by $G=(G^{\afa\beta})$ the corresponding Gram matrix
\cite{normalform, GFM}.
We represent the gradients of these periods in the orthonomal frame as follows:
\begin{equation}
  \nabla p_\afa = \sum_{i=1}^{n}\phi_{i\afa}f_i,
\end{equation}
where $\phi_{i\afa}=\phi_{i\afa}(u;\lmd)$,
then we have
\begin{equation}\label{ss. gram ortho}
  \phi_{i\afa}G^{\afa\beta}\phi_{j\beta}=\frac{\delta_{ij}}{u_i-\lmd},\quad i, j=1,\dots,n.
\end{equation}

\section{The solution of the loop equation}
In this section, we first recall briefly the derivation of the loop equation of a generalized Frobenius manifold $M$, and then give the proof of uniqueness and existence of solution of this equation by following the approach of \cite{normalform} for a usual semisimple Frobenius manifold.

\subsection{Derivation of the loop equation}
Recall that the tau cover of the Principal Hierarchy of the generalized Frobenius manifold $M$ has the form \cite{GFM}
\begin{equation}\label{tau-cover}
        \pfrac{f}{t^{j,q}}=f_{j,q},\quad
    \pfrac{f_{i,p}}{t^{j,q}}=\Omg_{i,p;j,q},\quad
    \pfrac{v^\afa}{t^{j,q}}=\eta^{\al\gamma}\p_x\Omg_{\gamma,0;j,q}.
  \end{equation}
It has unknown functions $f, f_{i,p}, v^\afa$, where
\begin{equation}\label{zh-23}
\alpha\in\{1,2\dots,n\},\quad (i, p)\in   \mcalI= \Big(\{1,\dots,n\}\times\bbZ_{\geq 0}\Big)\cup
  \Big(\{0\}\times\bbZ\Big),
\end{equation}
and $\Bigset{\Omg_{I,J}}{I,J\in\mcalI}$ are the 2-point functions on $M$.
Since the flow $\frac{\p}{\p t^{0,0}}$ is given by the translation along the spatial variable $x$, we will identify the time variable $t^{0,0}$ with $x$.

The tau cover \eqref{tau-cover} of the Principal Hierarchy of $M$ possesses a family of Virasoro symmetries $\left\{\frac{\p}{\p s_m}\mid m\ge -1\right\}$ of the following form \cite{GFM}:
\begin{align}
  \pfrac{f}{s_m}&=
    a_m^{I,J}f_If_J+b_{m;J}^It^Jf_I+\hat c_{m;I,J}t^It^J, \label{vir sym 1}\\
  \pfrac{f_I}{s_m}&=
    \pp{t^I}\pfrac{f}{s_m}=(2a_m^{K,L}\Omg_{I,K}+b^L_{m;I})f_L
    +(b^K_{m;L}\Omg_{I,K}+2\hat c_{m;I,L})t^L, \label{vir sym 2}\\
  \pfrac{v^\afa}{s_m}
&=
  \eta^{\afa\beta}\p_x\pp{t^{\beta,0}}\pfrac{f}{s_m}, \label{vir sym 3}
\end{align}
where $a_m^{I,J}$, $b^I_{m;J}$ and $\hat c_{m;I,J}$ are certain constants which only depend on the monodromy data $\eta,\mu,R$ of $M$ at $z=0$,
and they are called the \textit{extended Virasoro coefficients} of $M$.
Here and in what follows summation over repeated upper and lower capital Latin indices with range $\mcalI$ is assumed.
Note that the coefficients $a_m^{I,J}$, $\hat c_{m; I,J}$ and the 2-point functions $\Omg_{I,J}$
satisfy the following symmetry properties
\begin{equation}\label{sym of a,c,Omg}
a_m^{I,J}=a_m^{J,I},\quad
\hat c_{m;I,J} = \hat c_{m;J,I},\quad
\Omg_{I,J}=\Omg_{J,I}
\end{equation}
for all $m\geq -1$ and $I,J\in\mcalI$.
Each flow $\pp{s_m}\,(m\geq -1)$ commutes with all the flows $\pp{t^I}$, i.e.,
\begin{equation} \label{[sm pp tI]=0}
\left[
  \pp{s_m},\pp{t^I}
\right]=0,\quad m\geq -1,\, I\in\mcalI.
\end{equation}

Introduce the following $\bbC$-algebra
\begin{equation}
  \mathscr A
=C^\infty(J^\infty(M))
 [
    f,f_I
 ][\![t^I]\!],
\end{equation}
here $f$ and $f_I, t^I
\,(I\in\mcalI)$ are independent formal variables.
Then $\pp{t^I}\,(I\in\mcalI)$ can be regraded as a derivation on $\mathscr A$ in a natural way,
i.e., $\pp{t^I}\in\Der(\mathscr A)$ acts on $f$, $f_J$ and $v^\afa$ as in \eqref{tau-cover}, and
\[
  \pfrac{\mcalF}{t^I}=\sum_{s\geq 0}\pfrac{\mcalF}{v^{\afa,s}}\p_x^s\pfrac{v^\afa}{t^I}
\]
for $\mcalF\in C^\infty(J^\infty(M))$.
Similarly, $\pp{s_m}\,(m\geq -1)$ can also be regarded as elements in
$\Der(\mathscr A)$ by assigning $\pfrac{f}{s_m}$,
$\pfrac{f_I}{s_m}$, $\pfrac{v^\afa}{s_m}$ as in \eqref{vir sym 1}--\eqref{vir sym 3}, and
\begin{equation} \label{p tI p sm}
  \pfrac{\mcalF}{s_m}=\sum_{s\geq 0}\pfrac{\mcalF}{v^{\afa,s}}\p_x^s\pfrac{v^\afa}{s_m},\quad
  \pfrac{t^I}{s_m}=0
\end{equation}
for $\mcalF\in C^\infty(J^\infty(M))$ and $I\in\mcalI$. In this sense,
the relations $\left[\pp{t^I},\pp{t^J}\right]=\left[\pp{t^I},\pp{s_m}\right]=0$ still hold true.

The Virasoro symmetries can be represented in terms of
a family of Virasoro operators $L_m\,(m\geq -1)$
which are constructed for an arbitrary generalized Frobenius manifold in \cite{GFM}, and have the following form:
\begin{align}
  L_m(\bft,\pp{\bft}) &=\, a_m^{I,J}\frac{\p^2}{\p t^I\p t^J}
  +b_{m;I}^Jt^I\pp{t^J}
  +\hat c_{m;I,J}t^It^J
  +\frac{\delta_{m,0}}{4}
     \mathrm{tr}\left(\frac14-\mu^2\right).  \label{Virasoro operator 202308}
\end{align}
The coefficients $a_m^{I,J}$, $b^J_{m;I}$, $\hat c_{m;I,J}$
in the above formulae coincide with that of $\pp{s_m}$ given in \eqref{vir sym 1}.
These operators satisfy the Virasoro commutation relations
\begin{equation} \label{Vir commu for Lm}
  \left[L_m, L_{k}\right]=(m-k)L_{m+k},\quad \forall m, k\geq -1.
\end{equation}

We also note that $L_m\,(m\geq -1)$ are operators acting on the algebra $\mathscr A$,
and in this sense
\begin{equation}\label{[Lm, sm]=0}
\left[
  L_m , \pp{s_{k}}
\right]=0,\quad \forall\, m,k\geq -1
\end{equation}
because of \eqref{[sm pp tI]=0} and \eqref{p tI p sm}.

Consider a quasi-Miura transformation for the Principal Hierarchy of the form
\begin{equation}\label{quasi-Miura transf}
  v^\afa\mapsto w^\afa=v^\afa+\veps^2\eta^{\afa\gamma}\p_x\p_{t^{\gamma,0}}\Delta\mcalF,
\end{equation}
where
\begin{equation}\label{Delta mcalF 2308}
  \Delta\mcalF=\sum_{k\geq 1}\veps^{k-2}\mcalF^{[k]},
\end{equation}
and $\mcalF^{[k]}=\mcalF^{[k]}(v,v_x,v_{xx},...,v^{(N_k)})$ are functions on the jet space $J^\infty(M)$. Denote
\begin{equation}
  \mcalF= \veps^{-2}f+\Delta\mcalF \in\mathscr A[\veps^{-1}][\![\veps]\!],
\end{equation}
we say that the quasi-Miura transformation \eqref{quasi-Miura transf}
linearizes the Virasoro symmetries of the Principal Hierarchy if the actions of $\pp{s_m}$ on the tau function
\begin{equation}
  \tau= \exp(\mcalF)
\end{equation}
of the deformed hierarchy are given by
\begin{equation}\label{linearization condition}
  \frac{\p\tau}{\p s_m}=L_m\Bigl(\veps^{-1}\bft,\veps\pp{\bft}\Bigr)\tau,\quad m\geq -1,
\end{equation}
where the Virasoro operators $L_m$ are defined as in \eqref{Virasoro operator 202308}.
Introduce the generating series
\[
  \pp{s}=\sum_{m\geq -1}\frac{1}{\lmd^{m+2}}\pp{s_m},\quad
  L\Bigl(\veps^{-1}\bft,\veps\pp\bft;\lmd\Bigr)=
  \sum_{m\geq -1}
    \frac{1}{\lmd^{m+2}}L_m\Bigl(\veps^{-1}\bft,\veps\pp\bft\Bigr),
\]
then the linearization condition \eqref{linearization condition} can be rewritten as
\begin{equation}\label{linearization condition II}
  \frac{\p\tau}{\p s}=L\Bigl(\veps^{-1}\bft,\veps\pp{\bft};\lmd\Bigr)\tau,
\end{equation}
which is required to hold true identically with respect to $\lmd$.
By a straightforward calculation, we obtain
\begin{align}
L_m\Bigl(\veps^{-1}\bft,\pp{\bft}\Bigr)\tau
=&
  \left[
    \mcalD_m\Delta\mcalF+\veps^{-2}\pfrac{f}{s_m} +a_m^{I,J}\Omg_{I,J}
    +\veps^2 a_m^{I,J}\right. \notag \\
    &\quad
  \left.
      \left(
        \frac{\p^2\Delta\mcalF}{\p t^I\p t^J}
       +\pfrac{\Delta\mcalF}{t^I}\pfrac{\Delta\mcalF}{t^J}
      \right)
      +\frac{\delta_{m,0}}{4}\mathrm{tr}
          \left(\frac 14-\mu^2\right)
  \right]\tau,  \label{230831-1644}
\end{align}
where the operators $\mcalD_m$ are defined by
\begin{equation}\label{def of mcalD m}
  \mcalD_m= 2a_m^{I,J}f_I\pp{t^J}+b^J_{m;I}t^I\pp{t^J} \in\Der(\mathscr A),\quad m\ge -1.
\end{equation}

We have the following proposition on the operators $\mcalD_m$.
\begin{prop}[\!\!\cite{GFM}]
\label{prop: vector field Km}
For a function $\Delta\mcalF$ of the form \eqref{Delta mcalF 2308}, the following identities hold true:
\begin{equation}
  \mcalD_m\Delta\mcalF=K_m \Delta\mcalF+\pfrac{\Delta\mcalF}{s_m},\quad m\geq -1,
\end{equation}
where the vector fields
\begin{equation}\label{vector field K_m}
K_m=\sum_{s\geq 0}K_m^{\gamma,s}\pp{v^{\gamma,s}}
\end{equation}
are defined by the following
generating series:
\begin{align}
  &K^{\gamma,s}=\sum_{m\geq -1}\frac{K_m^{\gamma,s}}{\lmd^{m+2}} \notag\\
 =&
        (s+1)\p_x^s\left(\frac{1}{E-\lmd e}\right)^\gamma
      +
    \sum_{k=1}^{s}
      k{s+1\choose k+1}
        \left(
          \p_x^{s-k}\p^\gamma p_\afa
        \right)
        G^{\afa\beta}
        \left(
          \pp{\lmd}
          \p_x^k p_\beta
        \right),  \label{short expression of K^gamma_s}
\end{align}
with $p_1,\dots,p_n$ being any given basis of periods of $M$
and $(G^{\afa\beta})$ being the associated Gram matrix.
\end{prop}

Due to the above proposition, \eqref{230831-1644} can be rewritten as
\begin{align}
  L_m\Bigl(\veps^{-1}\bft,\veps\pp{\bft}\Bigr)\tau=&
  \left[
  K_m\Delta\mcalF +a_m^{I,J}\Omg_{I,J}
    +\veps^2 a_m^{I,J}
      \left(
        \frac{\p^2\Delta\mcalF}{\p t^I\p t^J}
       +\pfrac{\Delta\mcalF}{t^I}\pfrac{\Delta\mcalF}{t^J}
      \right)
  \right. \notag \\
&\quad
\left.
    +\frac{\delta_{m,0}}{4}
      \mathrm{tr}\left(\frac14-\mu^2\right)
  \right]\tau +\frac{\p\tau}{\p s_m}, \label{Lm tau 2309}
\end{align}
so the linearization condition \eqref{linearization condition} is equivalent to the equations
\[
  K_m\Delta\mcalF +a_m^{I,J}\Omg_{I,J}
    +\veps^2 a_m^{I,J}
      \left(
        \frac{\p^2\Delta\mcalF}{\p t^I\p t^J}
       +\pfrac{\Delta\mcalF}{t^I}\pfrac{\Delta\mcalF}{t^J}
      \right)
+\frac{\delta_{m,0}}{4}
      \mathrm{tr}\left(\frac14-\mu^2\right) =0
      \]
 for $\Delta\mcalF$, here $m\ge -1$. Thus we arrive at the following theorem.
\begin{thm}[\!\!\cite{GFM}]
The linearization condition \eqref{linearization condition II}
for a generalized Frobenius manifold (not necessarily semisimple) is equivalent to the following loop equation:
\begin{align}
&
  \sum_{s\geq 0}
    \pfrac{\Delta\mcalF}{v^{\gamma,s}}
    (s+1)\p_x^s\left(\frac{1}{E-\lmd e}\right)^\gamma \notag \\
&+
      \sum_{s\geq 0}
    \pfrac{\Delta\mcalF}{v^{\gamma,s}}
    \sum_{k=1}^{s}
      k{s+1\choose k+1}
        \left(
          \p_x^{s-k}\p^\gamma p_\afa
        \right)
        G^{\afa\beta}
        \left(
          \pp{\lmd}
          \p_x^k p_\beta
        \right)
\notag\\
=&\,
 \frac{\veps^2}{2}
 \sum_{k,\ell\geq 0}
 \left(
   \pfrac{\Delta\mcalF}{v^{\afa,k}}
   \pfrac{\Delta\mcalF}{v^{\beta,\ell}}
  +\frac{\p^2\Delta\mcalF}{\p v^{\afa,k}\p v^{\beta,\ell}}
 \right)
 (\p_x^{k+1}\p^\afa p_\sigma)
 G^{\sigma\rho}
 (\p_x^{\ell+1}\p^\beta p_\rho)
\notag\\
&
 +\frac{\veps^2}{2}
  \sum_{k\geq 0}
  \pfrac{\Delta\mcalF}{v^{\afa,k}}
  \p_x^{k+1}
  \left[
    \nabla\pfrac{p_\sigma}{\lmd}
    \cdot
    \nabla\pfrac{p_\rho}{\lmd}
    \cdot v_x
  \right]^\afa
  G^{\sigma\rho}
\notag
\\&
+
 \frac12 G^{\afa\beta}
 \pfrac{p_\afa}{\lmd}*\pfrac{p_\beta}{\lmd}-
 \frac{1}{4\lmd^2}
 \mathrm{tr}
 \left(
   \frac14-\mu^2
 \right),\label{loop equation-2308}
\end{align}
where $p_1,\dots,p_n$ is any given basis of periods of $M$,
$(G^{\afa\beta})$ is the associated Gram matrix,
and the star product $*$ is defined as in \cite{normalform, GFM}.
\end{thm}

Let us proceed to prove the existence and uniqueness (up to the addition of constants) of solution of the above loop equation for a semisimple generalized Frobenius manifold $M$ of dimension $n$. We first give two useful lemmas.
\begin{lem}
The operators $\pp{s_m}\in\Der(\mathscr A)$ defined in \eqref{vir sym 1}--\eqref{vir sym 3}, \eqref{p tI p sm}
satisfy the following Virasoro commutation relations:
\begin{equation}\label{Vir commu for pp sm}
\left[\pp{s_m},\pp{s_{k}}\right]=(k-m)\pp{s_{m+k}},\quad m,k\geq -1.
\end{equation}
\end{lem}
\begin{proof}
Due to \eqref{tau-cover},\eqref{[sm pp tI]=0} and \eqref{p tI p sm},
it suffices to verify the validity of the relations
\[\left[\pp{s_m},\pp{s_{k}}\right]f=(k-m)\pp{s_{m+k}}f,\quad m, k\ge -1.\]
Note that the operators $L_m$
satisfy the Virasoro commutation relations \eqref{Vir commu for Lm},
which are equivalent to the following identities:
\begin{align}
&(m-k)a_{m+k}^{I,J}=
  2\left(a_m^{I,K}b_{k;K}^J-a_{k}^{I,K}b_{m;K}^J\right), \label{Vir coef eqn 1}\\
&(m-k)b_{m+k;J}^I =
  \left(
    b_{m;J}^Kb_{k;K}^I-b_{k;J}^Kb_{m;K}^I
  \right)
 +4\left(a_m^{I,K}\hat c_{k;K,J}-a_{k}^{I,K}\hat c_{m;K,J}\right), \label{Vir coef eqn 2}\\
&(m-k)\hat c_{m+k;I,J} =
  2\left(
    b_{m;I}^K\hat c_{k;K,J} - b_{k;I}^K\hat c_{m;K,J}
  \right),\label{Vir coef eqn 3}\\
&(m-k)\frac{\delta_{m+k,0}}{4}
  \mathrm{tr}
  \left(\frac14-\mu^2\right)
=
  2\left(
    a_m^{K,L}\hat c_{k;K,L} - a_{k}^{K,L}\hat c_{m;K,L}
  \right),
\end{align}
where $m,k\geq -1$ and $I,J\in\mcalI$.
Therefore, a straightforward calculation yields
\begin{align*}
\left[
  \pp{s_m},\pp{s_{k}}\right]f
&=
  \pp{s_m}\left(
    a_{k}^{I,J}f_If_J
   +b_{k;J}^I t^J f_I
  \right) - (m\leftrightarrow k)
\\
&=
  2a_{k}^{I,K}b_{m;K}^Jf_If_J
 +\left(
   b_{k;J}^K b_{m;K}^I
   +4a_{k}^{I,K}\hat c_{m;K,J}
  \right) t^J f_I
 +2 b_{k;J}^K\hat c_{m;I,K} t^I t^J \\
&\quad
  -(m\leftrightarrow k) \\
&=
  (k-m)
  \left(
    a_{m+k}^{I,J}f_If_J
   +b_{m+k;J}^I t^J f_I
   +\hat c_{m+k;I,J}t^I t^J
  \right) \\
&=
  (k-m)\pp{s_{m+k}}f.
\end{align*}
Hence the lemma is proved.
\end{proof}
Due to \eqref{Vir commu for pp sm}, the symmetries $\pp{s_m},\, m\geq -1$
are called the Virasoro symmetries of the tau cover \eqref{tau-cover} of the Principal Hierarchy of $M$.

\begin{lem}For an arbitrary generalized Frobenius manifold, the operators
\begin{equation}
  \mcalK_m = \mcalD_m -\pp{s_m},\quad m\geq -1
\end{equation}
of $\Der(\mathscr A)$ satisfy the following Virasoro commutation relations:
\begin{equation}
  \left[
    \mcalK_m, \mcalK_{\ell}
  \right] = (m-\ell)\mcalK_{m+\ell},\quad m,\ell\geq -1.
\end{equation}
\end{lem}
\begin{proof}
By a straightforward calculation we obtain
\begin{align*}
&\left[\mcalD_m,\mcalD_{\ell}\right]=
  \Big(
    4a_m^{I,J} a_{\ell}^{K,L}f_I\Omg_{J,K}
   -2a_m^{I,L}b_{\ell;K}^J t^K \Omg_{I,J}
\\&\qquad
    +2a_m^{I,J} b_{\ell;J}^L f_I
    +b_{m;I}^J b_{\ell;J}^L t^I
  \Big)\pp{t^L}-(m\leftrightarrow \ell),
\\
&\left[\mcalD_m, -\pp{s_{\ell}}\right]
+\left[-\pp{s_m},\mcalD_{\ell}\right]
=
\Big(
  4a_m^{I,L} a_{\ell}^{K,J}f_J\Omg_{I,K}
   +2a_m^{I,L}b_{\ell;K}^J t^K \Omg_{I,J}  \\
&\qquad
   +2a_m^{I,L}b_{\ell;I}^K f_K
   +4a_m^{I,L} \hat c_{\ell;I,J} t^J
\Big)\pp{t^L}-(m\leftrightarrow \ell).
\end{align*}
Thus from the relations
\eqref{sym of a,c,Omg},\eqref{Vir commu for pp sm} and \eqref{Vir coef eqn 1}--\eqref{Vir coef eqn 3}
it follows that
\begin{align*}
  \left[\mcalK_m,\mcalK_{\ell}\right]
&=
  \left[\mcalD_m,\mcalD_{\ell}\right] + \left[\mcalD_m, -\pp{s_{\ell}}\right]
+\left[-\pp{s_m},\mcalD_{\ell}\right] +\left[\pp{s_m}, \pp{s_{\ell}}\right] \\
&=
  2\left(
    a_m^{I,J} b_{\ell;J}^L - a_{\ell}^{I,J} b_{m;J}^L
  \right)f_I\pp{t^L} \\
&\quad
 +
  \left[
    \left(
      b_{m;I}^J b_{\ell;J}^L - b_{\ell;I}^J b_{m;J}^L
    \right)
   +4\left(
       a_m^{J,L} \hat c_{\ell;I,J} - a_{\ell}^{J,L} \hat c_{m;I,J}
     \right)
  \right] t^I\pp{t^L}\\
&\quad
  +(k-m)\pp{s_{m+\ell}} \\
&=
  (m-\ell)
  \left(
    2a_{m+\ell}^{I,J}f_I\pp{t^J}
   +b_{m+\ell;I}^Jt^I\pp{t^J} - \pp{s_{m+\ell}}
  \right) \\
&=(m-\ell)\mcalK_{m+\ell},
\end{align*}
the lemma is proved.
\end{proof}

From Proposition \ref{prop: vector field Km} we know that
the restrictions of the operators $\mcalK_m\in\Der(\mathscr A)$ to $C^\infty(J^\infty(M))$
are the vector fields $K_m$ defined in \eqref{vector field K_m},
therefore we obtain the following corollary.
\begin{cor}
  For a generalized Frobenius manifold
  the vector fields $K_m$ on the jet space $J^\infty(M)$ defined in \eqref{vector field K_m}
  satisfy the following Virasoro commutation relations:
  \begin{equation}\label{Vir commu for Km}
    \left[K_m, K_{\ell}\right] = (m-\ell)K_{m+\ell},\quad m,\ell\geq -1.
  \end{equation}
\end{cor}

\subsection{The uniqueness theorem}
In this subsection, we are to prove that the solution
$\Delta\mcalF=\sum_{k\geq 1}\veps^{k-2}\mcalF^{[k]}$
to the loop equation \eqref{loop equation-2308} for a semisimple generalized Frobenius manifold is unique up to the addition of constant terms.

\begin{lem}\label{lemma:poles of K^gamma_i}
Let $M$ be an $n$-dimensional semisimple generalized Frobenius manifold with canonical coordinates $u_1,\dots,u_n$ and flat coordinates $v^1,\dots,v^n$,
then the differential polynomial $K^{\gamma,s}$ defined in \eqref{short expression of K^gamma_s} is a rational function in $\lmd$
with poles of order $s+1$ at $\lmd=u_i$, $i=1,\dots,n$.
Moreover, the coefficient of the highest order pole $\frac{1}{(\lmd-u_i)^{s+1}}$ of $K^{\gamma,s}$ is equal to
\begin{equation}
  -\frac{(2s+1)!!}{2^s}\psi_i^\gamma\sqrt{h_i}(u_{i,x})^s,
\end{equation}
where $\psi_{i}^\afa=\psi_{i\beta}\eta^{\beta\afa}$ and $\sqrt{h_i}$ are defined in
\eqref{def of psi i afa} and \eqref{def of h_i} respectively.
\end{lem}

\begin{proof} Fix an $i\in\{1,2,...,n\}$, then from \eqref{ss. GM-1} and \eqref{ss. GM-2} we obtain
  \begin{align}
    \p_x^k\phi_{i\afa}
   &=\,
    \frac{(2k-1)!!}{2^k}
    \frac{(u_{i,x})^k}{(\lmd-u_i)^k}
   \left(
     \phi_{i\afa}+2\sum_{j=1}^{n}V_{ij}\phi_{j\afa}
   \right)
\notag\\
&\quad
   +\text{lower-degree terms of $\frac{1}{\lmd-u_i}$}.
\label{pole lemma}
  \end{align}
Note that
\begin{equation}
  \left(\frac{1}{E-\lmd e}\right)^\gamma = \frac{\sqrt{h_i}}{u_i-\lmd}\psi_i^\gamma,
\end{equation}
therefore we have
\begin{align}
  (s+1)\p_x^s
  \left(
    \frac{1}{E-\lmd e}
  \right)^\gamma
&=\,
  -\frac{\sqrt{h_i}\psi_i^\gamma(u_{i,x})^s}{(\lmd-u_i)^{s+1}}(s+1)!
  +\cdots,
\end{align}
where ``$\cdots$'' stands for the lower-degree terms of $\frac{1}{\lmd-u_i}$.
Using \eqref{ss. GM-3} and \eqref{pole lemma}, we also obtain
\begin{align}
  \p_x^{s-k}\p^\gamma p_\afa
&=
  \p_x^{s-k}\left(\psi_i^\gamma\phi_{i\afa}\right)
=
  \psi_i^\gamma\p_x^{s-k}\phi_{i\afa}+\cdots \notag\\
&=
  \frac{\psi_i^\gamma(u_{i,x})^{s-k}}{(\lmd-u_i)^{s-k}}\cdot
  \frac{(2s-2k-1)!!}{2^{s-k}}
  \left(
    \phi_{i\afa}+2\sum_{j=1}^{n}V_{ij}\phi_{j\afa}
  \right)+\cdots, \\
\pp\lmd\p_x^kp_\beta &=
  \pp\lmd\p_x^{k-1}
    \left(\sqrt{h_i}u_{i,x}\phi_{i\beta}\right) \notag\\
&=
  \sqrt{h_i}u_{i,x}
  \pp{\lmd}
  \left(
    \frac{(2k-3)!!}{2^{k-1}}\cdot\frac{(u_{i,x})^{k-1}}{(\lmd-u_i)^{k-1}}
    \left(
      \phi_{i\beta}+2\sum_{j=1}^{n}V_{ij}\phi_{j\beta}
    \right)
  \right)+\cdots \notag\\
&=
  -\frac{\sqrt{h_i}(u_{i,x})^k}{(\lmd-u_i)^k}
   \frac{(2k-1)!!}{2^k}
   \left(
     \phi_{i\beta}+2\sum_{j=1}^{n}V_{ij}\phi_{j\beta}
   \right)+\cdots .
\end{align}
Then by using \eqref{ss. gram ortho} and \eqref{short expression of K^gamma_s} we arrive at
\begin{align}
  K^{\gamma,s}=&
    -\frac{\psi^\gamma_i\sqrt{h_i}(u_{i,x})^s}{(\lmd-u_i)^{s+1}}
    \left(
      (s+1)!-\frac{1}{2^s}\sum_{k=1}^{s}k{s+1\choose k+1}(2s-2k-1)!!(2k-1)!!
    \right) \notag \\
&
  +\text{lower-degree terms of $\frac{1}{\lmd-u_i}$}.
\end{align}
Thus, in order to prove the lemma we only need to verify the validity of the following identities for all $s\geq 0$:
\begin{equation}\label{comb identity-230815}
  (s+1)!-\frac{1}{2^s}\sum_{k=1}^{s}k{s+1\choose k+1}(2s-2k-1)!!(2k-1)!!
 =
 \frac{(2s+1)!!}{2^s}.
\end{equation}
In fact, the validity of these identities follows from the relations
\begin{align}
&
  \sum_{s=0}^{\infty}
  \left[
    (s+1)!-\frac{1}{2^s}\sum_{k=1}^{s}k{s+1\choose k+1}(2s-2k-1)!!(2k-1)!!
  \right]\frac{z^s}{s!}  \notag \\
=&
  \frac{1}{(1-z)^2}
 -2\frac{\td}{\td z}
   \sum_{s=0}^{\infty}
   \sum_{k=0}^{s}
     \frac{k(2k-1)!!}{(k+1)!}\left(\frac z2\right)^{k+1}
     \cdot
     \frac{(2s-2k-1)!!}{(s-k)!}\left(\frac z2\right)^{s-k} \notag\\
=&
  \frac{1}{(1-z)^2}
  -2\frac{\td}{\td z}
  \left[
    \left(
      \sum_{k=0}^{\infty}\frac{k(2k-1)!!}{(k+1)!}\left(\frac z2\right)^{k+1}
    \right)
    \left(
      \sum_{k=0}^{\infty}
        \frac{(2k-1)!!}{k!}
        \left(\frac z2\right)^k
    \right)
  \right]\notag \\
=&
  \frac{1}{(1-z)^2}-2\frac{\td}{\td z}
  \left(
    \frac{2-2\sqrt{1-z}-z}{2\sqrt{1-z}}\cdot\frac{1}{\sqrt{1-z}}
  \right) \notag\\
=&
  \frac{1}{(1-z)^{\frac 32}}
=
 \sum_{s=0}^{\infty}
   \frac{(2s+1)!!}{2^s}\cdot\frac{z^s}{s!}.\notag
\end{align}
The lemma is proved.
\end{proof}

\begin{lem} \label{lem: a uniqueness lemma}
Let $M$ be an $n$-dimensional semisimple generalized Frobenius manifold, and
$\Bigset{B_{\gamma,s}}{1\leq\gamma\leq n,\, 0\leq s\leq N}$
be a family of functions on the jet space $J^\infty(M)$,
where $N$ is a certain positive integer.
If
\begin{equation}\label{230823-2152}
  \sum_{s=0}^{N}
    B_{\gamma,s}
    K^{\gamma,s}=0
\end{equation}
holds true identically in $\lmd$, then all the functions $B_{\gamma,s}$ are equal to zero.
In particular, if a function $\mcalF=\mcalF(v,v_x,...,v^{(N)})$ on the jet space $J^\infty(M)$ satisfies the equation
\[
 \sum_{s=0}^{N}K^{\gamma,s}\pfrac{\mcalF}{v^{\gamma,s}} = 0,
\]
then $\mcalF\equiv\mathrm{const}$.
\end{lem}

\begin{proof}
According to Lemma \ref{lemma:poles of K^gamma_i}, the left-hand side of \eqref{230823-2152}, which we denote by $Q$, is a rational function of $\lmd$. It has a pole of order $N+1$ at $\lmd=u_i$, and
\begin{equation}
 \res_{\lmd=u_i} (\lmd-u_i)^N Q= -B_{\gamma,N}
  \psi_i^\gamma
  \frac{(2N+1)!!}{2^N}\sqrt{h_i}(u_{i,x})^N=0,\quad i=1,\dots,n.
\end{equation}
Due to the nondegeneracy of the matrix $(\psi_i^\gamma)$ we arrive at
\[
  B_{\gamma,N}
  \equiv 0,\quad \forall 1\leq \gamma\leq n.
\]
By continuing this procedure recursively, we obtain
\[B_{\gamma,s}\equiv 0,\quad s=0,\dots,N, \gamma=1,\dots,n.\]
The lemma is proved.
\end{proof}
\begin{cor}\label{cor: the form of Delta F}
If $\Delta\mcalF=\sum_{k\geq 1}\veps^{k-2}\mcalF^{[k]}$ given in \eqref{Delta mcalF 2308}
is a solution to the loop equation \eqref{loop equation-2308} of a semisimple generalized Frobenius manifold,
then $\mcalF^{[2k-1]}$ must be a constant function for each integer $k\geq 1$.
Therefore, the solution $\Delta\mcalF$ (if exists) must be of the following form up to the addition of constant terms:
\begin{equation}\label{genus expansion of Delta F}
  \Delta\mcalF=\sum_{g\geq 1}\veps^{2g-2}\mcalF_g,
\end{equation}
where $\mcalF_g=\mcalF_g(v,v_x,...,v^{(N'_g)})$ for certain integers $N'_g$, $g\geq 1$.
\end{cor}
\begin{proof}
We prove the corollary by induction on $k$. For the case $k=1$,
  the coefficients of $\veps^{-1}$ in both sides of the loop equation \eqref{loop equation-2308} give us the relation
  \[
    \sum_{s=0}^{N_1}\pfrac{\mcalF^{[1]}}{v^{\gamma,s}}K^{\gamma,s}=0,
  \]
  where $\mcalF^{[1]}=\mcalF^{[1]}(v,v_x,...,v^{(N_1)})$.
  Therefore it follows from Lemma \ref{lem: a uniqueness lemma} that $\mcalF^{[1]}$ is a constant function.
Now let $k\geq 2$, and suppose $\mcalF^{[2k'-1]}$ are constant functions for all $1\leq k'\leq k-1$,
  then the coefficients of $\veps^{2k-3}$ in both sides of \eqref{loop equation-2308} lead to the relation
\[\sum_{s=0}^{N_{2k-1}}\pfrac{\mcalF^{[2k-1]}}{v^{\gamma,s}}K^{\gamma,s}=0,\]
  therefore $\mcalF^{[2k-1]}$ must also be a constant function. The corollary is proved.
\end{proof}
\begin{thm}\label{thm:the uniqueness}
For a semisimple generalized Frobenius manifold, the solution \eqref{genus expansion of Delta F} to
the loop equation \eqref{loop equation-2308} is unique up to the addition of constant terms.
\end{thm}
\begin{proof}
  It suffices to show that each $\mcalF_g$ is uniquely determined by the loop equation \eqref{loop equation-2308},
  which is a simple corollary of Lemma \ref{lem: a uniqueness lemma}.
\end{proof}

The solution $\mcalF_g$ in \eqref{genus expansion of Delta F} is called the
\textit{genus $g$ free energy} of the generalized Frobenius manifold.
It turns out that the integer $N'_g$ in Corollary \ref{cor: the form of Delta F} can be chosen as $3g-2$,
which will be proved in the following two subsections.

\subsection{The genus one free energy}

Now we proceed to show the existence of a solution of the form \eqref{genus expansion of Delta F} to
the loop equation \eqref{loop equation-2308} of a semisimple generalized Frobenius manifold.
To this end, we first solve the genus one free energy $\mcalF_1$ in terms of the canonical coordinates.
Comparing the $\veps^0$-coefficients in \eqref{loop equation-2308}, we obtain
\begin{equation}\label{loop equation for F1}
    \sum_{s=0}^{N_1'}\pfrac{\mcalF_1}{v^{\gamma,s}}K^{\gamma,s}
  =
    \frac12 G^{\afa\beta}
 \pfrac{p_\afa}{\lmd}*\pfrac{p_\beta}{\lmd}-
 \frac{1}{4\lmd^2}
 \mathrm{tr}
 \left(
   \frac14-\mu^2
 \right).
  \end{equation}

\begin{lem}
For a semisimple generalized Frobenius manifold and an arbitrary basis of its periods $p_1(v,\lmd),\dots, p_n(v,\lmd)$
with the associated Gram matrix $(G^{\afa\beta})$, the following identity holds true
\begin{align}\label{s.s. star product}
\pfrac{p_\afa}{\lmd}*\pfrac{p_\beta}{\lmd}G^{\afa\beta}
&=\,-\frac18\sum_{i=1}^{n}\frac{1}{(\lmd-u_i)^2}
-\sum_{i<j}\frac{V_{ij}^2}{(\lmd-u_i)(\lmd-u_j)}
+\frac{1}{2\lmd^2}\mathrm{tr}\left(\frac14-\mu^2\right).
\end{align}
\end{lem}
\begin{proof}
Recall that the formula
\begin{equation}
  \nabla(f*g)=\nabla f\cdot\nabla g
\end{equation}
for the star product $*$ holds true \cite{normalform, GFM} for any smooth functions $f(v;\lmd)$ and $g(v;\lmd)$,
where $\nabla$ is the gradient with respect to the flat metric $\eta$,
and $``\cdot"$ is the multiplication of the Frobenius algebra.
Then by \eqref{230821-1052} we have
\begin{align*}
 & \nabla\left(\pfrac{p_\afa}{\lmd}*\pfrac{p_\beta}{\lmd}G^{\afa\beta}\right)
=
  \left(\nabla\pfrac{p_\afa}{\lmd}\right)\cdot
  \left(\nabla\pfrac{p_\beta}{\lmd}\right)G^{\afa\beta}\\
=&
  \sum_{i,j=1}^{n}
    \frac{1}{h_ih_j}
      \left(
        \pp{u_i}\pfrac{p_\afa}{\lmd}
      \right)
      \left(
        \pp{u_j}\pfrac{p_\beta}{\lmd}
      \right)G^{\afa\beta}
      \pp{u_i}\cdot\pp{u_j}\\
=&
  \sum_{i=1}^{n}
    \frac{1}{h_i}
      \left(
      \pfrac{\phi_{i\afa}}{\lmd}G^{\afa\beta}\pfrac{\phi_{i\beta}}{\lmd}
      \right)\pp{u_i}.
\end{align*}
On the other hand,
\[
\nabla\left(\pfrac{p_\afa}{\lmd}*\pfrac{p_\beta}{\lmd}G^{\afa\beta}\right)
= \sum_{i=1}^{n}
  \frac{1}{h_i}
  \pp{u_i}
  \left(
    \pfrac{p_\afa}{\lmd}*\pfrac{p_\beta}{\lmd}G^{\afa\beta}
  \right)\pp{u_i}.
\]
Therefore, by using the relations \eqref{ss. GM-3}, \eqref{ss. gram ortho} we obtain
\begin{align*}
 & \pp{u_i}
  \left(
    \pfrac{p_\afa}{\lmd}*\pfrac{p_\beta}{\lmd}G^{\afa\beta}
  \right)\\
=&
  \frac{1}{(\lmd-u_i)^2}
  \left(
    \frac12\phi_{i\afa}
   +\sum_{j\neq i}V_{ij}\phi_{j\afa}
  \right)G^{\afa\beta}
  \left(
    \frac12\phi_{i\beta}
   +\sum_{k\neq i}V_{ik}\phi_{k\beta}
  \right) \\
=&
  \frac{1}{(\lmd-u_i)^2}
  \left(
    \frac14\frac{1}{u_i-\lmd}
   +\sum_{j\neq i}\frac{V_{ij}^2}{u_j-\lmd}
  \right)\\
=&
  -\frac14\frac{1}{(\lmd-u_i)^3}
  -\sum_{j\neq i}\frac{V_{ij}^2}{(\lmd-u_i)^2(\lmd-u_j)}.
\end{align*}
Then due to \eqref{zh-1-3}, there exists a function $c=c(\lmd)$ such that
\begin{equation}
  \pfrac{p_\afa}{\lmd}*\pfrac{p_\beta}{\lmd}G^{\afa\beta}
 =
  -\frac18\sum_{i=1}^{n}\frac{1}{(\lmd-u_i)^2}
  -\sum_{i<j}\frac{V_{ij}^2}{(\lmd-u_i)(\lmd-u_j)}+c(\lmd).
\end{equation}
On the other hand, from \cite{GFM} we know that the left-hand side of the above formula can be represented as
\[
  \pfrac{p_\afa}{\lmd}
  *\pfrac{p_\beta}{\lmd}G^{\afa\beta}= -2\sum_{(i,p),(j,q)\in\mcalI}
  a^{i,p;j,q}(\lmd)\Omega_{i,p;j,q}(v),
\]
where
\[a^{i,p;j,q}(\lmd)=\sum_{m\geq -1}\frac{1}{\lmd^{m+2}}a_m^{i,p;j,q}\]
are the generating function of Virasoro coefficients, and $\Omega_{i,p;j,q}(v)$ are the two-point functions of the generalized Frobenius manifold.
Note that $a_{-1}^{i,p;j,q}=a_{0}^{i,p;j,q}=0$,
therefore
\[
  \pfrac{p_\afa}{\lmd}*\pfrac{p_\beta}{\lmd}G^{\afa\beta} = O\left(\frac{1}{\lmd^3}\right)\quad \text{when $\lmd\to\infty$},
\]
hence we obtain
\[
  c(\lmd)= \frac n{8\lmd^2}+\sum_{i<j}\frac{V_{ij}^2}{\lmd^2}
  =\frac{1}{2\lmd^2}\mathrm{tr}\left(\frac14-V^2\right)
  =\frac{1}{2\lmd^2}\mathrm{tr}\left(\frac14-\mu^2\right).
\]
The lemma is proved.
\end{proof}

\begin{rmk}
  The formula \eqref{s.s. star product} can be rewritten as
\begin{equation}\label{230821-1227}
  \pfrac{p_\afa}{\lmd}*\pfrac{p_\beta}{\lmd}G^{\afa\beta} =
-\frac18\mathrm{tr}
   \left(
     \frac{1}{\mcalU-\lmd I}
   \right)^2
+\frac12\mathrm{tr}
  \left(
    \frac{1}{\mcalU-\lmd I}\mu
  \right)^2
+\frac{1}{2\lmd^2}\mathrm{tr}\left(\frac14-\mu^2\right),
\end{equation}
which is canonical coordinate-free.
We conjecture that it holds true for a generalized Frobenius manifold without the semisimplicity condition.
\end{rmk}

\begin{thm}
The solution of the equation \eqref{loop equation for F1} is given by
\begin{equation}\label{genus one free energy}
  \mcalF_1=\mcalF_1(u,u_x)=
    \log\frac{\tau_I(u)}{J^{1/24}(u)}+\frac{1}{24}
    \sum_{i=1}^{n}\log u_{i,x},
\end{equation}
where $\tau_I$ is the isomonodromic tau function defined by \eqref{isomonodromic tau function}, and
\begin{equation}
  J(u)
 =\prod_{i=1}^{n}\sqrt{h_i}\,.
\end{equation}
\end{thm}

\begin{proof}
Due to Theorem \ref{thm:the uniqueness},
it suffices to solve the equation \eqref{loop equation for F1} under the assumption that $\mcalF_1$ only depends on $v$ and $v_x$, i.e. $\mcalF_1=\mcalF_1(v,v_x)$.
Then the equation \eqref{loop equation for F1} can be written as
\begin{equation}\label{loop equation for F1-new}
    \pfrac{\mcalF_1}{v^\gamma}K^{\gamma,0} +
    \pfrac{\mcalF_1}{v^\gamma_x}K^{\gamma,1}
  =
    -\frac{1}{16}\sum_{i=1}^{n}\frac{1}{(\lmd-u_i)^2}
    -\frac12\sum_{i<j}\frac{V_{ij}^2}{(\lmd-u_i)(\lmd-u_j)},
\end{equation}
where $K^{\gamma,0}$ and $K^{\gamma,1}$ are defined in \eqref{short expression of K^gamma_s}.
Note that
\begin{equation*}
  K^{\gamma,0} =
  \left(
    \frac{1}{E-\lmd e}
  \right)^\gamma
=
  \sum_{j=1}^{n}\sqrt{h_j}\psi_j^\gamma\frac{1}{u_j-\lmd},
\end{equation*}
then from \eqref{ortho of Psi},  \eqref{zh-1-4}, \eqref{p_i psi} and \eqref{def of rotation coef}
it follows that
\begin{align}
\pfrac{\mcalF_1}{v^\gamma}K^{\gamma,0}
=&
  \sum_{i=1}^{n}
    \left(
      \pfrac{\mcalF_1}{u_i}
      \pfrac{u_i}{v^\gamma}
     +\pfrac{\mcalF_1}{u_{i,x}}
      \p_x\left(\pfrac{u_i}{v^\gamma}\right)
    \right)
  \sum_{j=1}^{n}
    \left(
      \sqrt{h_j}\psi_{j}^\gamma\frac{1}{u_j-\lmd}
    \right)
\notag\\
=&
  -\sum_{i=1}^{n}\frac{1}{\lmd-u_i}\pfrac{\mcalF_1}{u_i}
  +\sum_{i=1}^{n}\frac{1}{\lmd-u_i}
     \sum_{j=1}^{n}
       \pfrac{\mcalF_1}{u_{j,x}}
       \sqrt{h_i}\psi_i^\gamma
       \frac{\psi_{j\gamma}}{(\sqrt{h_j})^2}
       \p_x\sqrt{h_j}
\notag\\
&
  -\sum_{i=1}^{n}\frac{1}{\lmd-u_i}
   \sum_{j=1}^{n}
     \pfrac{\mcalF_1}{u_{j,x}}
       \sqrt{h_i}\psi_i^\gamma
       \frac{\p_x(\psi_{j\gamma})}{\sqrt{h_j}}
\notag\\
=&
  \sum_{i=1}^{n}\frac{1}{\lmd-u_i}
  \left(
    -\pfrac{\mcalF_1}{u_i}
    +\pfrac{\mcalF_1}{u_{i,x}}\p_x\log\sqrt{h_i}
          \phantom{\sum_{k\neq j}^n}
  \right.
\notag\\
&
  \left.
    -\sum_{j=1}^{n}
      \pfrac{\mcalF_1}{u_{j,x}}\sqrt{h_i}\psi_i^\gamma\frac{1}{\sqrt{h_j}}
      \sum_{k\neq j}
        \gamma_{kj}\psi_{k\gamma}(u_{k,x}-u_{j,x})
  \right)
\notag\\
=&
  \sum_{i=1}^{n}\frac{1}{\lmd-u_i}
  \left(
    -\pfrac{\mcalF_1}{u_i}
    +\sum_{j=1}^{n}u_{j,x}\pfrac{\mcalF_1}{u_{j,x}}
     \pfrac{\log\sqrt{h_j}}{u_i}
  \right.
\notag\\
&
  \left.
    +\sum_{j\neq i}
     \left(
       \pfrac{\mcalF_1}{u_{i,x}}
       \pfrac{\log\sqrt{h_i}}{u_j}u_{j,x}
      -\pfrac{\mcalF_1}{u_{j,x}}
       \pfrac{\log\sqrt{h_j}}{u_i}u_{i,x}
     \right)
  \right).   \label{v gamma K 0 gamma}
\end{align}
Similarly, from \eqref{230821-1052}, \eqref{ss. GM-3} and \eqref{ss. gram ortho} it follows that
\begin{align}
K^{\gamma,1}=&
  2\p_x\left(\frac{1}{E-\lmd e}\right)^\gamma
  +(\p^\gamma p_\afa)G^{\afa\beta}
  \left(\pp\lmd\p_x p_\beta\right)
\notag\\
&=
  2\sum_{i=1}^{n}\p_x
   \left(
     \sqrt{h_i}\psi_i^\gamma\frac{1}{u_i-\lmd}
   \right)
  -\frac12\sum_{i,j=1}^{n}
    \frac{\psi_i^\gamma\sqrt{h_j}u_{j,x}}{\lmd-u_j}
    \phi_{i\afa}G^{\afa\beta}
    \left(
      \phi_{j\beta}+2\sum_{k\neq j}V_{jk}\phi_{k\beta}
    \right)
\notag\\
&=
  2\sum_{i=1}^{n}
    \p_x\left(
      \sqrt{h_i}\psi_i^\gamma
    \right)\frac{1}{u_i-\lmd}
 -\frac 32\sum_{i=1}^{n}
   \frac{\psi_i^\gamma\sqrt{h_i}u_{i,x}}{(\lmd-u_i)^2}
 -\sum_{i=1}^n\sum_{j\neq i}
    \frac{\psi_i^\gamma V_{ij}\sqrt{h_j}u_{j,x}}{(\lmd-u_i)(\lmd-u_j)}. \notag
\end{align}
Then by using the relations \eqref{zh-1-5}, \eqref{zh-1-6}
and $\pfrac{u_{i,x}}{v^\gamma_x}=\pfrac{u_i}{v^\gamma}=\frac{\psi_{i\gamma}}{\sqrt{h_i}}$ we obtain
\begin{align}
  \pfrac{\mcalF_1}{v_x^\gamma}K^{\gamma,1}
=&
  2\sum_{i,j=1}^{n}
    \pfrac{\mcalF_1}{u_{i,x}}
    \frac{\psi_{i\gamma}}{\sqrt{h_i}}
    \p_x\left(\sqrt{h_j}\psi_j^\gamma\right)\frac{1}{u_j-\lmd}
  -\frac 32\sum_{i=1}^{n}
   \frac{u_{i,x}}{(\lmd-u_i)^2}\pfrac{\mcalF_1}{u_{i,x}}
\notag\\
&
  -\sum_{i=1}^{n}
     \sum_{j\neq i}
       \frac{V_{ij}}{(\lmd-u_i)(\lmd-u_j)}
       \frac{\sqrt{h_j}}{\sqrt{h_i}}
       \pfrac{\mcalF_1}{u_{i,x}}u_{j,x}
\notag\\
=&
  -2\sum_{i=1}^{n}
    \frac{1}{\lmd-u_i}
    \sum_{j=1}^{n}
      \pfrac{\mcalF_1}{u_{j,x}}
      \frac{\psi_{j\gamma}}{\sqrt{h_j}}
      \p_x\left(\sqrt{h_i}\psi_i^\gamma\right)
  -\frac32\sum_{i=1}^{n}
    \frac{u_{i,x}}{(\lmd-u_i)^2}\pfrac{\mcalF_1}{u_{i,x}}
\notag\\
&
  +\sum_{i=1}^{n}\frac{1}{\lmd-u_i}
   \sum_{j\neq i}
     \left(
       \pfrac{\mcalF_1}{u_{i,x}}
       \pfrac{\log\sqrt{u_i}}{u_j}u_{j,x}
      -\pfrac{\mcalF_1}{u_{j,x}}
       \pfrac{\log\sqrt{u_j}}{u_i}u_{i,x}
     \right). \notag
\end{align}
It can be verified that for each $1\leq i\leq n$ we have
\begin{align*}
&
  \sum_{j=1}^{n}
    \pfrac{\mcalF_1}{u_{j,x}}
      \frac{\psi_{j\gamma}}{\sqrt{h_j}}
      \p_x\left(\sqrt{h_i}\psi_i^\gamma\right)
\\
=&
  \sum_{j=1}^{n}
    \pfrac{\mcalF_1}{u_{j,x}}
    \pfrac{\log\sqrt{h_j}}{u_i}u_{j,x}
+\sum_{j\neq i}
  \left(
    \pfrac{\mcalF_1}{u_{i,x}}
    \pfrac{\log\sqrt{h_i}}{u_j}u_{j,x}
   -\pfrac{\mcalF_1}{u_{j,x}}
    \pfrac{\log\sqrt{h_j}}{u_i}u_{i,x}
  \right),
\end{align*}
therefore
\begin{align}
\pfrac{\mcalF_1}{v_x^\gamma}K^{\gamma,1}
=&
  -\frac32\sum_{i=1}^{n}
    \frac{u_{i,x}}{(\lmd-u_i)^2}
    \pfrac{\mcalF_1}{u_{i,x}}
  -2\sum_{i=1}^{n}
    \frac{1}{\lmd-u_i}
    \sum_{j=1}^{n}
      u_{j,x}\pfrac{\mcalF_1}{u_{j,x}}
      \pfrac{\log\sqrt{h_j}}{u_i}
\notag\\
&
  -\sum_{i=1}^{n}\frac{1}{\lmd-u_i}
   \sum_{j\neq i}
    \left(
    \pfrac{\mcalF_1}{u_{i,x}}
    \pfrac{\log\sqrt{h_i}}{u_j}u_{j,x}
   -\pfrac{\mcalF_1}{u_{j,x}}
    \pfrac{\log\sqrt{h_j}}{u_i}u_{i,x}
    \right).  \label{v gamma s K 1 gamma}
\end{align}
Due to \eqref{v gamma K 0 gamma} and \eqref{v gamma s K 1 gamma},
the equation \eqref{loop equation for F1-new} is equivalent to
\begin{align}
&\,
  \frac32\sum_{i=1}^{n}
   \frac{u_{i,x}}{(\lmd-u_i)^2}
     \pfrac{\mcalF_1}{u_{i,x}}
  +\sum_{i=1}^{n}\frac{1}{\lmd-u_i}
   \left(
     \pfrac{\mcalF_1}{u_i}
    +\sum_{j=1}^{n}
      u_{j,x}\pfrac{\mcalF_1}{u_{j,x}}
      \pfrac{\log\sqrt{h_j}}{u_i}
   \right)
\notag\\
=&\,
  \frac{1}{16}\sum_{i=1}^{n}\frac{1}{(\lmd-u_i)^2}
  +\sum_{i=1}^{n}\frac{1}{\lmd-u_i}
   \left(
     \frac12\sum_{j\neq i}\frac{V_{ij}^2}{u_i-u_j}
   \right).
\end{align}
Comparing the coefficients of $\frac{1}{(\lmd-u_i)^2}$ and $\frac{1}{\lmd-u_i}$, we arrive at
\begin{align}
  \pfrac{\mcalF_1}{u_{i,x}}
&=\,
  \frac{1}{24}\frac{1}{u_{i,x}}\\ 
  \pfrac{\mcalF_1}{u_i}
&=\,
  -\sum_{j=1}^{n}u_{j,x}\pfrac{\mcalF_1}{u_{j,x}}
   \pfrac{\log\sqrt{h_j}}{u_i}
  +\frac12\sum_{j\neq i}\frac{V_{ij}}{u_i-u_j} \notag\\
&=\,
  -\frac{1}{24}\pp{u_i}
  \left(
    \sum_{j=1}^{n}\log\sqrt{h_j}
  \right)+H_i.
\end{align}
Therefore $\mcalF_1$ is of the form \eqref{genus one free energy}.
The theorem is proved.
\end{proof}

\subsection{The existence theorem}

We are now to prove the existence of a solution to the loop equation \eqref{loop equation-2308} in the form
\[
  \Delta\mcalF=\sum_{g\geq 1}\veps^{2g-2}\mcalF_g(v,v_x,...,v^{(3g-2)}).
\]
We are to prove the existence of $\mcalF_g$ by induction on $g$.
The genus one free energy $\mcalF_1$ is already given in \eqref{genus one free energy}.
For $g\geq 2$, suppose $\mcalF_{g'}$ has already been solved for each $1\leq g'\leq g-1$,
then from the $\veps^{2g-2}$-coefficients of \eqref{loop equation-2308} we know that $\mcalF_g$ satisfies the following equation:
\begin{align}
  \sum_{s=0}^{3g-2}
    \pfrac{\mcalF_g}{v^{\gamma,s}}K^{\gamma,s}
  =&
    \frac12\sum_{k,l=0}^{3g-5}
      \left(
        \pp{v^{\afa,k}}
        \left(
        \pfrac{\mcalF_{g-1}}{v^{\beta,l}}
        \right)
       \right)
    (\p_x^{k+1}\p^{\afa} p_{\omg})G^{\omg\rho}(\p_x^{l+1}\p^\beta p_\rho)
\notag\\
&+\frac12
        \sum_{{g',g''\geq 1\atop g'+g''=g}}
        \sum_{k=0}^{3g'-2}
        \sum_{l=0}^{3g''-2}
          \pfrac{\mcalF_{g'}}{v^{\afa,k}}
          \pfrac{\mcalF_{g''}}{v^{\beta,l}}
(\p_x^{k+1}\p^{\afa} p_{\sigma})G^{\sigma\rho}(\p_x^{l+1}\p^\beta p_\rho)
\notag\\
&
  +\frac12\sum_{k=0}^{3g-5}\pfrac{\mcalF_{g-1}}{v^{\afa,k}}
   \p_x^{k+1}
     \left[
       \nabla\pfrac{p_\sigma}{\lmd}
       \cdot
       \nabla\pfrac{p_\rho}{\lmd}\cdot v_x
     \right]^{\afa}G^{\sigma\rho}.  \label{linear system of H g afa s}
\end{align}
The above equation can be regarded as a linear system for the unknowns
\begin{equation}
  H_{g;\gamma,s}:=\pfrac{\mcalF_g}{v^{\gamma,s}},\quad 1\leq\gamma\leq n,\, 0\leq s\leq 3g-2.
\end{equation}
We are to show the existence of $H_{g;\gamma,s}$,
and then prove that the solutions $H_{g;\gamma,s}$ satisfy the following compatibility condition:
\begin{equation}\label{compatibility condition for H gas}
  \pfrac{H_{g;\afa,r}}{v^{\beta,s}}
=\pfrac{H_{g;\beta,s}}{v^{\afa,r}}.
\end{equation}

\begin{lem}\label{lemma: the H_g;gamma s}
The linear system \eqref{linear system of H g afa s} for the unknowns
$H_{g;\gamma,s}$ have a unique solution
\[H_{g;\gamma,s}\in C^{\infty}(M)\left[u_x,\frac{1}{u_x};u_{xx},u_{xxx},...,u^{(3g-s-1)}\right]\]
for $g\ge 1, 1\le\gamma\le n, 0\le s\le 3g-2$.
\end{lem}
\begin{proof}
We prove the lemma by induction on $g$. The existence of a unique solution of \eqref{linear system of H g afa s} for the $g=1$ case is already proved in the last subsection.
Now let $g\geq 2$, and suppose that the linear systems \eqref{linear system of H g afa s} for all $1\le g'<g$ have already been solved,
  then the linear system \eqref{linear system of H g afa s}
  of the unknowns $H_{g;\gamma,s}, 1\leq\gamma\leq n,\, 0\leq s\leq 3g-2$ reads
\begin{align}
&\sum_{s=0}^{3g-2}
  H_{g;\gamma,s}K^{\gamma,s}
  =
    \frac12\sum_{k,l=0}^{3g-5}
      \pfrac{H_{g-1;\beta,l}}{v^{\afa,k}}
    (\p_x^{k+1}\p^{\afa} p_{\sigma})G^{\sigma\rho}(\p_x^{l+1}\p^\beta p_\rho)
\notag\\
&
     \quad +\frac12
        \sum_{{g',g''\geq 1\atop g'+g''=g}}
        \sum_{k=0}^{3g'-2}
        \sum_{l=0}^{3g''-2}
          H_{g';\afa,k}
          H_{g'';\beta,l}
(\p_x^{k+1}\p^{\afa} p_{\sigma})G^{\sigma\rho}(\p_x^{l+1}\p^\beta p_\rho)
\notag\\
&
 \quad  +\frac12\sum_{k=0}^{3g-5}
   H_{g-1;\afa,k}
   \p_x^{k+1}
     \left[
       \nabla\pfrac{p_\sigma}{\lmd}
       \cdot
       \nabla\pfrac{p_\rho}{\lmd}\cdot v_x
     \right]^{\afa}G^{\sigma\rho}.  \label{new linear system}
\end{align}
Both sides of \eqref{new linear system} are meromorphic functions in $\lmd$
which vanish at $\lmd\to\infty$,
and all the possible poles of them are at $u_1,\dots,u_n$.
Due to Lemma \ref{lemma:poles of K^gamma_i},
the functions $K^{\gamma,s}$ can be written in the form
\begin{equation}
K^{\gamma,s}=
  \sum_{i=1}^{n}
    \sum_{k=1}^{s+1}
      \frac{Q^{\gamma,s}_{i,k}}{(\lmd-u_i)^k},
\end{equation}
where $Q^{\gamma,s}_{i,k}$ are functions on the jet space.
Using \eqref{ss. GM-1}--\eqref{ss. GM-3} and \eqref{ss. gram ortho},
it can be verified that
\begin{equation}\label{the space containing Q}
  Q^{\gamma,s}_{i,k}\in C^{\infty}(M)\left[u_x,u_{xx},...,u^{(s-k+2)}\right]
\end{equation}
by a straightforward counting of the degrees of $\frac{1}{\lmd-u_i}$ and the jet variables. Thus we have
\begin{equation}\label{230824-1657}
  \sum_{s=0}^{3g-2}H_{g;\gamma,s}K^{\gamma,s}
=
  \sum_{i=1}^{n}\sum_{k=0}^{3g-2}
    \frac{1}{(\lmd-u_i)^{k+1}}
      \sum_{s=k}^{3g-2}
        H_{g;\gamma,s}Q^{\gamma,s}_{i,k+1}.
\end{equation}
By using the following formulae of change of coordinates on the jet spaces \cite{normalform}:
\begin{equation}\label{change jet variables}
  \pfrac{u_i^{(p)}}{v^{\gamma,r}}=
  {p\choose r}\p_x^{p-r}\pfrac{u_i}{v^\gamma}
\end{equation}
and by using the induction hypothesis we have
\begin{equation}
  \pfrac{H_{g-1;\beta,l}}{v^{\afa,k}}=0\quad\text{if $k+l\geq 3g-3$}.
\end{equation}
Then by \eqref{ss. GM-1}--\eqref{ss. GM-3} and \eqref{ss. gram ortho},
the order of poles $\lmd=u_i$ in the right-hand side of \eqref{new linear system} is at most $3g-1$, therefore
\begin{equation}\label{230824-1658}
\text{RHS of \eqref{new linear system}}=
  \sum_{i=1}^{n}\sum_{k=1}^{3g-1}
    \frac{P_{g;i,k}}{(\lmd-u_i)^k}
\end{equation}
for some functions $P_{g;i,k}$ on the jet space $J^\infty(M)$.
It follows from our induction hypothesis that
\begin{equation}\label{the space containing P}
  P_{g;i,k}\in C^\infty(M)
    \left[
      u_x,\frac{1}{u_x};u_{xx},u_{xxx},...,u^{(3g-k)}
    \right].
\end{equation}

Due to \eqref{230824-1657}, \eqref{230824-1658},
the equation \eqref{new linear system} for the unknowns $H_{g;\gamma,s}$ can be rewritten as
\begin{equation}
  \sum_{i=1}^{n}
  \sum_{k=0}^{3g-2}
    \frac{1}{(\lmd-u_i)^{k+1}}
    \left(
      \sum_{s=k}^{3g-2}H_{g;\gamma,s}Q^{\gamma,s}_{i,k+1} - P_{g;i,k+1}
    \right)=0,
\end{equation}
therefore it
is equivalent to the linear system
\begin{equation}\label{linear system III}
  \sum_{s=k}^{3g-2}H_{g;\gamma,s}Q^{\gamma,s}_{i,k+1}=P_{g;i,k+1},\quad 1\leq i\leq n,\, 0\leq k\leq 3g-2.
\end{equation}
Introduce the $1\times n$ row vectors
\begin{equation*}
  \bsH_{g;s}:=(H_{g;1,s},...,H_{g;n,s}),\quad
  \bsP_{g;k}:=(P_{g;1,k},...,P_{g;n,k}),
\end{equation*}
and the $n\times n$ matrices $\bsQ^s_{k}$ with entries
\begin{equation*}
  (\bsQ^s_{k})_i^\gamma := Q_{i,k}^{\gamma,s},
\end{equation*}
then the system \eqref{linear system III} can be represented in the form
\begin{equation}
\left(
  \bsH_{g;3g-2}, \bsH_{g;3g-3},...,\bsH_{g;0}
\right) \mcalQ_g =
 \left(
  \bsP_{g;3g-1}, \bsP_{g;3g-2},...,\bsP_{g;1}
\right),
\end{equation}
where
\begin{equation}
  \mcalQ_g :=
  \begin{pmatrix}
    \bsQ^{3g-2}_{3g-1} & \bsQ^{3g-2}_{3g-2} & \cdots & \bsQ^{3g-2}_{1} \\
      & \bsQ^{3g-3}_{3g-2} & \cdots & \bsQ^{3g-3}_{1} \\
      &   & \ddots & \vdots \\
      &   &  & \bsQ^{0}_{1}
  \end{pmatrix}
\end{equation}
is a block upper triangular matrix. Due to Lemma \ref{lemma:poles of K^gamma_i},
the diagonal blocks of $\mcalQ_g$ are given by
\begin{equation}
(\bsQ^{k}_{k+1})_i^\gamma = -\frac{(2k+1)!!}{2^k}(u_{i,x})^k\sqrt{h_i}\psi^\gamma_i,\quad 0\leq k\leq 3g-2,
\end{equation}
therefore $\bsQ^k_{k+1}$ is invertible,
and all entries of $(\bsQ^k_{k+1})^{-1}$ belong to $C^\infty(M)[\frac{1}{u_x}]$
for $0\leq k\leq 3g-2$. Hence $\mcalQ_g$ is invertible,
and the unknowns $H_{g;\gamma,s}$
are uniquely determined by
\[
  \left(
  \bsH_{g;3g-2}, \bsH_{g;3g-3},...,\bsH_{g;0}
\right) :=
 \left(
  \bsP_{g;3g-1}, \bsP_{g;3g-2},...,\bsP_{g;1}
\right)\mcalQ_g^{-1}.
\]
Moreover, from \eqref{the space containing Q},\eqref{the space containing P}
we know that
\[H_{g;\gamma,s}\in C^{\infty}(M)\left[u_x,\frac{1}{u_x};u_{xx},u_{xxx},...,u^{(3g-s-1)}\right]\]
for all $1\leq\gamma\leq n$ and $0\leq s\leq 3g-2$.
The lemma is proved.
\end{proof}

Now let us proceed to prove the compatibility condition \eqref{compatibility condition for H gas}.
We first rewrite
the identity \eqref{Lm tau 2309} as follows:
\begin{equation} \label{230903-1526}
  L_m(\veps^{-1}\bft,\veps\pp{\bft})\tau =
  \left(
    \sum_{g\geq 1}
    \mcalA_{m,g}\veps^{2g-2}
  \right)\tau +\pp{s_m}\tau,
\end{equation}
where $\mcalA_{m,g}\in C^\infty(J^\infty(M))$ are of the form
\begin{equation}\label{def of fai mg}
  \mcalA_{m,g} = K_m\mcalF_g+\fai_{m,g}
\end{equation}
with
\[\fai_{m,g}=\fai_{m,g}(\mcalF_1,\mcalF_2,...,\mcalF_{g-1})\in C^\infty(J^\infty(M)),\]
i.e., $\fai_{m,g}$ only depends on the lower-genus terms $\mcalF_{g'}\, (g'\leq g-1)$.
More precisely, $-\fai_{m,g}$ equals the right-hand side of \eqref{linear system of H g afa s} for $g\geq 2$,
and the loop equation \eqref{linear system of H g afa s} for $\mcalF_g$ is equivalent to
\begin{equation}
  \mcalA_{m,g}=0,\quad m\geq -1.
\end{equation}
\begin{lem}
Let $M$ be an arbitrary generalized Frobenius manifold, and $g\geq 2$. Suppose there are functions $\mcalF_1,\mcalF_2,...,\mcalF_{g-1}$
satisfying $\mcalA_{m,g'}=0$ for $1\leq g'\leq g-1$ and $m\geq -1$.
Then the functions $\fai_{m,g}, m\geq -1$ defined by \eqref{def of fai mg} satisfy the relations
\begin{equation}\label{Km fai mg}
K_m\fai_{m',g}-K_{m'}\fai_{m,g}=(m-m')\fai_{m+m',g}, \quad m,m'\geq -1,
\end{equation}
where the vector fields $K_m$ on the jet space $J^\infty(M)$ are defined by \eqref{vector field K_m}.
\end{lem}
\begin{proof}
Since $\mcalA_{m,g'}=0$ for $g'\leq g-1$, the identity \eqref{230903-1526} for
\[
  \Delta\mcalF=\sum_{k=1}^{g-1}\veps^{2k-2}\mcalF_k, \quad \tau=\exp\left(\veps^{-2}f+\Delta\mcalF\right)
\]
(i.e., choose $\mcalF_k\equiv 0$ for all $k\geq g$) has the form
\[
  L_{m'}^{[\veps]}\tau =
  \left(
    \sum_{k\geq g}\mcalA_{m',k}\veps^{2k-2}
  \right)\tau + \frac{\p\tau}{\p s_{m'}},\quad m'\geq -1,
\]
here we use the short notation $L_m^{[\veps]}:=L_m\left(\veps^{-1}\bft,\veps\pp{\bft}\right)$.
By using \eqref{[Lm, sm]=0}, \eqref{def of mcalD m} and Proposition \ref{prop: vector field Km} we obtain
\begin{align*}
L_m^{[\veps]}L_{m'}^{[\veps]}\tau =&
  \left(
    \sum_{k\geq g}\mcalA_{m',k}\veps^{2k-2}
  \right)\left(L_m^{[\veps]}\tau\right)
 +\left(
    \sum_{k\geq g}a_m^{I,J}\frac{\p^2\mcalA_{m',k}}{\p t^I\p t^J}\veps^{2k}
  \right)\tau
\\&
 +2\left(
    \sum_{k\geq g}a_m^{I,J}\pfrac{\mcalA_{m',k}}{t^I}\veps^{2k}
   \right)\left(
     \veps^{-2}f_J+\sum_{k\geq 1}\veps^{2k-2}\pfrac{\mcalF_k}{t^J}
   \right)\tau
\\&
  +\sum_{k\geq g}
    \left(
      b_{m;J}^I\pfrac{\mcalA_{m',k}}{t^I}t^J\veps^{2k-2}
    \right)\tau
  +\pp{s_{m'}}\left(L_m^{[\veps]}\tau\right)
\\=&
  \left(
    \sum_{k\geq g}\mcalA_{m',k}\veps^{2k-2}
  \right)\left(
    \sum_{k\geq g}\mcalA_{m,k}\veps^{2k-2}
  \right)\tau
 +\left(\sum_{k\geq g}\mcalA_{m',k}\veps^{2k-2}\right)\pp{s_m}\tau
\\&
 +\Big(\mcalD_m\mcalA_{m',g}\veps^{2g-2}+O(\veps^{2g})\Big)\tau
 +\pp{s_{m'}}
  \left(
    \left(
      \sum_{k\geq g}\mcalA_{m,k}\veps^{2k-2}
    \right)\tau +\pp{s_m}\tau
  \right)\\
=&
  \left(
    \sum_{k\geq g}\mcalA_{m',k}\veps^{2k-2}
  \right)\left(
    \sum_{k\geq g}\mcalA_{m,k}\veps^{2k-2}
  \right)\tau
\\&
  +\left(\sum_{k\geq g}\mcalA_{m',k}\veps^{2k-2}\right)\pp{s_m}\tau
  +\left(\sum_{k\geq g}\mcalA_{m,k}\veps^{2k-2}\right)\pp{s_{m'}}\tau
  \\ &+\left(
     \pp{s_m}\mcalA_{m',g}+\pp{s_{m'}}\mcalA_{m,g}
   \right)\tau
  +\left(
    K_m\mcalA_{m',g}\veps^{2g-2}+O(\veps^{2g})
   \right)\tau + \pp{s_{m'}}\pp{s_m}\tau,
\end{align*}
here $m,m'\geq -1$. Note that we choose $\mcalF_g\equiv 0$ here,
therefore \eqref{def of fai mg} yields $\mcalA_{m,g}=\fai_{m,g}$ for $m\geq -1$.
Thus from \eqref{Vir commu for pp sm} we obtain
\begin{align*}
  \left[
    L_{m}^{[\veps]}, L_{m'}^{[\veps]}
  \right]\tau &=\,
  \Big(
    (K_m\fai_{m',g}-K_{m'}\fai_{m,g})\veps^{2g-2}+O(\veps^{2g})
  \Big)\tau + (m-m')\pp{s_{m+m'}}\tau.
\end{align*}
On the other hand,
\begin{align*}
  \left[
    L_{m}^{[\veps]}, L_{m'}^{[\veps]}
  \right]\tau
&=
  (m-m')L^{[\veps]}_{m+m'}\tau \\
&=
 \left(
   (m-m')\fai_{m+m',g}\veps^{2g-2}+O(\veps^{2g})
 \right)\tau + (m-m')\pp{s_{m+m'}}\tau,
\end{align*}
hence the lemma is proved.
\end{proof}

\begin{thm}\label{thm-existence}
 Let $M$ be a semisimple generalized Frobenius manifold, then the functions $H_{g;\gamma,s}$ on $J^\infty(M)$
 determined by Lemma \ref{lemma: the H_g;gamma s} satisfy the relations
 \begin{equation}\label{230903-1811}
  \pfrac{H_{g;\afa,r}}{v^{\beta,s}}
=\pfrac{H_{g;\beta,s}}{v^{\afa,r}}, \quad
 g\geq 1,\, 1\le\al,\beta\le n,\,0\leq r,s\leq 3g-2,
\end{equation}
therefore, there exists functions $\mcalF_g=\mcalF_g(v;v_x,...,v^{(3g-2)})$ such that
\[H_{g;\gamma,s}=\pfrac{\mcalF_g}{v^{\gamma,s}},\quad g\geq 1,\, 1\leq \gamma\leq n,\, 0\leq s\leq 3g-2,\]
and $\Delta\mcalF=\sum_{g\geq 1}\veps^{2g-2}\mcalF_g$ is a solution to the loop equation \eqref{loop equation-2308}.
\end{thm}
\begin{proof}
We prove the theorem by induction on $g$. In the $g=1$ case, the genus one free energy
  $\mcalF_1$ is given by \eqref{genus one free energy}.
  For $g\geq 2$, suppose the functions $\mcalF_1,\mcalF_2,...,\mcalF_{g-1}$ satisfying the equations $\mcalA_{m,k}=0$
  have already been found. Note that the equation \eqref{new linear system} for $H_{g;\gamma,s}$ is equivalent to
  \begin{equation}\label{230903-1741}
    \sum_{s\geq 0}H_{g;\gamma,s}K_m^{\gamma,s}+\fai_{m,g}=0,\quad m\geq -1,
  \end{equation}
  here we denote $H_{g;\gamma,s}:=0$ for $s>3g-2$.
  Applying the vector field $K_{m'}$ to both sides of \eqref{230903-1741} we arrive at
  \[
    \sum_{r,s\geq 0}
      K_m^{\afa,r}K_{m'}^{\beta,s}
      \pfrac{H_{g;\afa,r}}{v^{\beta,s}}
   +\sum_{r,s\geq 0}
     H_{g;\afa,r}\pfrac{K_m^{\afa,r}}{v^{\beta,s}}K_{m'}^{\beta,s}
   +K_{m'}\fai_{m,g}=0.
  \]
Subtracting the above identity by the one that is obtained from it by interchanging $m$ and $m'$,
and using the relations \eqref{Vir commu for Km}, \eqref{Km fai mg} and \eqref{230903-1741}, we obtain
\begin{equation}\label{230903-1805}
  \sum_{r,s=0}^{3g-2}
    K_m^{\afa,r}K_{m'}^{\beta,s}
      \left(
        \pfrac{H_{g;\afa,r}}{v^{\beta,s}}
       -\pfrac{H_{g;\beta,s}}{v^{\afa,r}}
      \right)=0,\quad m,m'\geq -1.
\end{equation}
Now by applying Lemma \ref{lem: a uniqueness lemma} to \eqref{230903-1805} we arrive at \eqref{230903-1811},
hence the theorem is proved.
\end{proof}

\begin{rmk}
The proofs of the uniqueness and existence of solution to the loop equation of a semisimple generalized Frobenius manifold given in this section are analogues of the ones given in \cite{normalform} for a semisimple Frobenius manifold. Note that the proof of existence of solution to the loop equation of a semisimple Frobenius manifold is given in a new 2005-version of \cite{normalform}.
\end{rmk}

\section{Example: the extended $q$-deformed KdV hierarchy}
In this section we consider a 1-dimensional generalized Frobenius manifold $M$ with potential
\begin{equation}\label{F=v^4/12}
  F=\frac{v^4}{12},
\end{equation}
here $v$ is the flat coordinate with respect to the metric $\eta=1$. The unit vector field and the Euler vector field are given by \cite{GFM}
\[
e=\frac1{2 v}\p_v,\quad E=\frac12 v\p_v.
\]
This generalized Frobenius manifold is semisimple with charge $d=1$ and with monodromy data $\mu=R=0$.

The Principal Hierarchy of $M$ defined in \cite{GFM} can be represented in the form
\begin{align}
  \pfrac{v}{t^{1,p}}  &= \p_x\p_v\theta_{1,p+1}, \quad p\geq 0,  \label{PH-t1p}\\
  \pfrac{v}{t^{0,-p}} &= \p_x\p_v\theta_{0,-p+1}, \quad p\geq 1, \label{PH-t0-p}\\
  \pfrac{v}{t^{0,p}}  &= \p_x\p_v\theta_{0,p+1}, \quad p\geq 0   \label{PH-t0p},
\end{align}
where the Hamiltonian densities $\theta_{\al,p}$ are given by
\begin{align*}
&\theta_{0,0}=\frac12 \log v,\quad \theta_{1,p}=\frac1{p!}\frac{v^{2p+1}}{2p+1},\quad p\ge 0,\\
&\theta_{0,p}=\frac{2^{p-1}}{(2p-1)!!}\frac{v^{2p}}{2p},\quad
\theta_{0,-p}=(-1)^{p+1}\frac{(2p-1)!!}{2^{p+1}}\frac{v^{-2p}}{2p},\quad p\ge 1.
\end{align*}
In this section, we are to prove that the topological deformation of the Principal Hierarchy of $M$ is equivalent to the extended $q$-deformed KdV hierarchy that is to be defined below. More precisely, let us assume that
\begin{equation}\label{240118-1435}
\Delta\mcalF(\veps)=\sum_{g\geq 1}\veps^{2g-2}\mcalF_g(v,v_x,...,v^{(3g-2)})
\end{equation}
is the solution
to the loop equation of $M$, and $v=v(t)$ is a solution to the Principal Hierarchy \eqref{PH-t1p}--\eqref{PH-t0p},
then we are to prove that the 2-point function
\[
  U=\frac{\Lmd-\Lmd^{-1}}{2\epsilon\p_x}
     \left(
       v+\veps^2\frac{\p^2\Delta\mcalF(\veps)}{\p x\p t^{1,0}}
     \right)
\]
satisfies the following integrable hierarchy which we call \textit{the extended $q$-deformed KdV hierarchy}:
\begin{align}
    \epsilon\pfrac{L}{t^{1,p}} &= c_{1,p}\left[(L^{p+\frac{1}{2}})_+,L\right],\quad  p\geq 0 \label{positive flow},\\
    \epsilon\pfrac{L}{t^{0,-p}} &= c_{0,-p}\left[(M^p)_-,L\right],\quad p\geq 1,\label{negative flow}\\
    \epsilon\pfrac{L}{t^{0,p}} &=c_{0,p}\left[(L^p\log L)_+,L\right], \quad p\geq 0.  \label{log flow}
\end{align}
Here the parameter $\epsilon$ and the shift operator $\Lmd$ are given by
\[\Lmd=\rme^{\epsilon\p_x}=\rme^{\sqrt{2}\,\rmi\veps\p_x},\quad \epsilon=\sqrt{2}\,\rmi\veps,\]
the Lax operator $L$ has the form
\begin{equation}\label{zh-1-7}
  L=\Lmd^2+U\Lmd,
\end{equation}
and the constants $c_{1,p}, c_{0,-p}, c_{0,p}$ are defined by
\begin{align}
  &c_{1,p} =\, (-1)^{p+1}\frac{2^{3p+2}}{(2p+1)!!},\quad c_{0,p}=(-1)^{p+1}\frac{2^{2p-1}}{p!},\quad p\geq 0, \label{zh-1-11}\\
  &c_{0,-p} = -\frac{(p-1)!}{4^p},\quad p\geq 1.  \label{zh-1-12}
         \end{align}
The fractional powers $L^{p+\frac{1}{2}}$ and the inverse powers $M^p=(L^{-1})^p$ of $L$ can be obtained from the operators
\begin{equation}\label{zh-1-9}
L^{\frac12}=\Lmd+\sum_{k\geq 0}a_{-k}\Lmd^{-k},\quad
M=\sum_{k\geq -1}b_k\Lmd^k,
\end{equation}
which satisfy the relations
\begin{equation}\label{zh-1-10}
\left(L^{\frac12}\right)^2=L,\quad LM=ML=1,
\end{equation}
and the operator $\log L$ will be defined in the Subsection\,\ref{zh-1-8}. Note that for any operator of the form $A=\sum_{p\in\bbZ}a_p\Lmd^p$,
we denote
\[A_+=\sum_{p\geq 0}a_p\Lmd^p,\quad A_-=\sum_{p<0}a_p\Lmd^p.\]

We will call the flows given in \eqref{positive flow} and \eqref{negative flow} the positive and negative flows of the extended $q$-deformed KdV hierarchy respectively. The first few flows of the above integrable hierarchy have the following forms:
\begin{align}
  \epsilon\pfrac{U}{t^{1,0}}&=
    4U\left(\frac{\Lmd-1}{\Lmd+1}U\right),\label{t1,0 flow}\\
  \epsilon\pfrac{U}{t^{1,1}}&=\,
    \frac{32}{3}
    \left(
      U\circ\frac{\Lmd-1}{\Lmd+1}\circ U\circ \frac{\Lmd}{\Lmd+1}
    \right)
    \left[
      \left(
        \frac{1}{\Lmd+1}U
      \right)^2
    \right],
 \label{t1,1 flow}\\
   \pfrac{U}{t^{0,0}} &=U_x,\quad
  \epsilon\pfrac{U}{t^{0,-1}} =\frac14\left(\frac{1}{U^+}-\frac{1}{U^-}\right), \label{t0,-1 flow}\\
  \epsilon\pfrac{U}{t^{0,-2}} &=
    -\frac 1{16}
  \left(
    \frac{1}{U^{++}U^+U^+}
   +\frac{1}{U^+U^+U}
   -\frac{1}{UU^-U^-}
   -\frac{1}{U^-U^-U^{--}}\right),
\end{align}
where $U^{\pm}=\Lambda^{\pm1} U$.

To prove the above assertion, we first show that the extended $q$-deformed KdV hierarchy possesses a bihamiltonian structure, and moreover its $\pp{t^{1,p}}$-flows ($p\ge 0$) and $\pp{t^{0,-p}}$-flows ($p\ge 1$)  coincide with the flows of the fractional Volterra hierarchy (FVH) \cite{DLYZ2,FVH} with a particular choice of its parameters $\sfp, \sfq, \sfr$.
On the other hand, we prove that the quasi-Miura transformation that is defined by the solution to the loop equation of $M$ coincides with the one that relates the fractional Volterra hierarchy with its dispersionless limit \cite{Crelles, JDG}. Then we arrive at the fact, after a certain Miura type transformation, that the $\pp{t^{1,p}}$-flows ($p\ge 0$) and the $\pp{t^{0,-p}}$-flows ($p\ge 1$) of the topological deformation of the Principal Hierarchy of $M$ coincide with the flows of the extended $q$-deformed KdV hierarchy given by \eqref{positive flow}, \eqref{negative flow}. Finally, by using the bihamiltonian property of the extended $q$-deformed KdV hierarchy, we show that the $\pp{t^{0,p}}$-flows ($p\ge 0$) of the topological deformation of the Principal Hierarchy of $M$ is equivalent to the flows \eqref{log flow} of the extended $q$-deformed KdV hierarchy.

\begin{rmk}
The positive flows of the $q$-deformed KdV hierarchy \eqref{positive flow} is introduced in \cite{Frenkel} in terms of the Lax operator
\[\hat{L}=D^2-t_1(z) D+1,\]
where the q-difference operator $D$ is defined by $[D\cdot f](z)=f(z q)$.
If we represent $z, q$ in the form $z=e^x, q=e^\ve$, then we are led to the
Lax operator of the form \eqref{zh-1-7} and the positive flows of the $q$-deformed KdV hierarchy.
The positive and negative flows of the $q$-deformed KdV hierarchy are also
studied in \cite{Boiti1, Boiti2, Kup, Shabat}.
\end{rmk}

\subsection{The bihamiltonian structure of the positive and negative flows}
Let us introduce the following operators:
\begin{equation}\label{biham structure of q-KdV}
  \mcalP_1= \frac1{2\epsilon}(\Lmd-\Lmd^{-1}),\quad
  \mcalP_2= \frac 2\epsilon U\frac{\Lmd-1}{\Lmd+1}U.
\end{equation}
It can be verified that $(\mcalP_1,\mcalP_2)$ is a bihamiltonian structure.

\begin{lem}
  The  flows \eqref{positive flow} and \eqref{negative flow} can be rewritten as
  \begin{align}
  \epsilon\pfrac{U}{t^{1,p}}
  &=
    c_{1,p}U(1-\Lmd)\res\left( L^{p+\frac12}\right), \label{positive flow - res form}  \\
    \epsilon\pfrac{U}{t^{0,-p}}
   &=
   c_{0,-p}(\Lmd^{-1}-\clm)\res\left(\Lmd M^p\right) \label{negative flow - res form}.
  \end{align}
\end{lem}
\begin{proof}
  From \eqref{positive flow} it follows that
  \begin{align*}
  &
    \frac{\epsilon}{c_{1,p}}\pfrac{U}{t^{1,p}}
  =
    \res\left(
      \left[
        (L^{p+\frac12})_+
      , L\right]\Lmd^{-1}
    \right) \\
  =&\,
    \res\left((L^{p+\frac12})_+(\Lmd+U)-(\Lmd^2+U\Lmd)(L^{p+\frac12})_+\Lmd^{-1}\right) \\
  =&\,
    U\res\left(L^{p+\frac12}-\Lmd L^{p+\frac12}\Lmd^{-1}\right)
  =
    U(1-\Lmd)\res\left( L^{p+\frac12}\right).
  \end{align*}
Similarly, from \eqref{negative flow} we obtain the relations
  \begin{align*}
  &
    \frac{\epsilon}{c_{0,-p}}
    \pfrac{U}{t^{0,-p}}
  =
    \res\left(
      \left[(M^p)_-,L\right]\clm^{-1}
    \right) \\
  =&\,
    \res\left((M^p)_-(\clm+U)\right)
   -\res\left((\clm^2+U\Lmd)(M^p)_-\Lmd^{-1}\right) \\
  =&\,
    \res\left(M^p\Lmd\right)-\res\left(\Lmd^2 M^p \Lmd^{-1}\right)
  =
    (\Lmd^{-1}-\Lmd)\res\left( \Lmd M^p\right).
  \end{align*}
Thus the lemma is proved.
\end{proof}

\begin{prop}
The positive and negative flows \eqref{positive flow}--\eqref{negative flow} are Hamiltonian systems with respect to the Poisson operator $\mcalP_1$, i.e., they can be represented in the form
\begin{align}
  \pfrac{U}{t^{1,p}}
&=
  \mcalP_1\frac{\delta H_{1,p}}{\delta U},\quad p\geq 0,  \label{positive flow - ham form}\\
  \pfrac{U}{t^{0,-p}}
&=
  \mcalP_1\frac{\delta H_{0,-p}}{\delta U},\quad p\geq 1, \label{negative flow - ham form}
\end{align}
where the densities of the Hamiltonians
\begin{align*}
H_{1,p}&=\int h_{1,p} \td x,\quad  p\geq 0, \\
H_{0,-p}&=\int h_{0,-p} \td x,\quad  p\geq 1
\end{align*}
are given by
\begin{align}
  & h_{1,p}= \frac{4c_{1,p}}{2p+3}\res\left(L^{p+\frac 32}\right), \quad p\geq 0,  \\
   &h_{0,-1}= \frac12
     \log V,\quad
   h_{0,-p}= \frac{2c_{0,-p}}{p-1}\res\left( M^{p-1} \right),\quad \text{$p\geq 2$}.\label{h0,-1 nega flow}
\end{align}
\end{prop}
\begin{proof}
By using the relations
\[
  \delta h_{1,p} \sim 2c_{1,p}\res\left(L^{p+\frac12}\delta L\right)
  =2c_{1,p}\res\left(L^{p+\frac12}\delta U\cdot \Lmd\right)
  \sim 2c_{1,p}\res\left(\delta U\cdot \Lmd L^{p+\frac12}\right)
\]
we obtain
\begin{equation} \label{var der H1p}
  \frac{\delta H_{1,p}}{\delta U} =2c_{1,p}\res\left(\Lmd L^{p+\frac12}\right),\quad p\geq 0.
\end{equation}
Here we define the equivalence relation $``\sim"$ as follows: a 1-form on the jet space is equivalent to zero if it is the $x$-derivative of another 1-form. Thus \eqref{positive flow - ham form} follows from \eqref{positive flow - res form}
and the relations
\begin{align*}
&\,
 U(1-\Lmd)\res\left(L^{p+\frac12}\right)
=
 \res\left(
   UL^{p+\frac12}-U\Lmd L^{p+\frac12}\Lmd^{-1}\right) \\
=&
 \res\left(
   L^{p+\frac12}U+(\Lmd^2-L) L^{p+\frac12}\Lmd^{-1}\right) \\
=&\,
  \res\left(
    L^{p+\frac12}U+\Lmd^2 L^{p+\frac12}\Lmd^{-1} -L^{p+\frac12}(\Lmd+U)
  \right) \\
=&
  (\Lmd-\Lmd^{-1})\res\left(\Lmd L^{p+\frac12}\right).
\end{align*}
Similarly, from the relations
  \begin{align*}
    \delta h_{0,-p}&=\, \frac{2 c_{0,-p}}{p-1}\res\left(\delta M^{p-1}\right)
    = 2c_{0,-p}\res\left(M^{p-2}\delta M\right) \\
  &\sim\,
    -2c_{0,-p}\res\left( M^{p-1}\delta U \cdot \Lmd M\right)
  \sim
    -2c_{0,-p}\res\left(\delta U\cdot \Lmd M^p\right)
  \end{align*}
  it follows that
  \begin{equation}\label{var der of H-p}
  \frac{\delta H_{0,-p}}{\delta U} = -2 c_{0,-p}
  \res\left(\Lmd M^p\right),\quad p\geq 1.
\end{equation}
Then by using \eqref{negative flow - res form} we arrive at the validity of \eqref{negative flow - ham form}.
The proposition is proved.
\end{proof}

\begin{prop}
The positive and negative flows \eqref{positive flow}, \eqref{negative flow} of the extended $q$-deformed KdV hierarchy are also Hamiltonian systems with respect to the second Hamiltonian operator $\mcalP_2$ defined in \eqref{biham structure of q-KdV}, i.e.,
\begin{align}
  \pfrac{U}{t^{1,p}}
&=\,
  \mcalP_2\frac{\delta G_{1,p}}{\delta U}, \quad p\geq 0, \label{positive flow - 2nd ham form}\\
  \pfrac{U}{t^{0,-p}}
&=\,
  \mcalP_2\frac{\delta G_{0,-p}}{\delta U}, \quad p\geq 1, \label{negative flow - 2nd ham form}
\end{align}
where the Hamiltonians are given by
\begin{align}
  G_{1,0}&=2\int U\td x, \
  G_{1,p}= \frac{2}{2p+1}H_{1,p-1}, \
  G_{0,-p}=
     -\frac 1p
     H_{0,-p-1} ,\quad p\geq 1. \label{second hamiltonian of neg flow}
\end{align}
\end{prop}
\begin{proof}
From the relations
\begin{align*}
\res\left(U\Lmd L^{p-\frac12}\right)
&=
  \res\left((L-\Lmd^2)L^{p-\frac12}\right)
=
  \res\left(L^{p+\frac12}\right)
 -\Lmd^2\res\left(L^{p-\frac12}\Lmd^2\right) \\
&=
  \res\left(L^{p+\frac12}\right)
 -\Lmd^2\res\left(L^{p-\frac12}(L-U\Lmd)\right)\\
&=
  (1-\Lmd^2)\res\left(L^{p+\frac12}\right)+\Lmd\res\left(U\Lmd L^{p-\frac12}\right)
\end{align*}
it follows the identity
\[
  (1-\Lmd)\res\left(U\Lmd L^{p-\frac12}\right)
 =(1-\Lmd^2)\res\left( L^{p+\frac12}\right).
\]
By applying the operator $U\frac{1}{1+\Lmd}$ to both sides of the above identity
and by using \eqref{positive flow - res form}, \eqref{var der H1p}
we arrive at \eqref{positive flow - 2nd ham form} for $p>0$.
For $p=0$, we can verify \eqref{positive flow - 2nd ham form} directly.

Similarly, from the relation
\begin{align*}
\res\left( U\Lmd M^{p+1} \right) &=\,
  \res\left((L-\Lmd^2)M^{p+1}\right)
=
  \res\left(M^{p+1}L-\Lmd^2M^{p+1}\right) \\
&=\,
  \res\left(
    M^{p+1}(\Lmd^2+U\Lmd) - \Lmd^2 M^{p+1}
  \right) \\
&=\,
  \res\left( M^{p+1}\Lmd^2\right)
 -\res\left(\Lmd^2 M^{p+1}\right)
 +\res\left(M^{p+1}U\Lmd\right) \\
&=\,
  \left(\Lmd-\Lmd^{-1}\right)
  \res\left(\Lmd M^{p+1}\Lmd\right)
 +\Lmd^{-1}\res\left(U\Lmd M^{p+1}\right)
\end{align*}
we obtain the identity
\[
  \left(1-\Lmd^{-1}\right)
  \res\left(
    U\Lmd M^{p+1}
  \right)
=
  \left(\Lmd^{-1}-\Lmd\right)
  \res\left(\Lmd M^{p+1}\Lmd\right).
\]
Applying $U\frac{\Lmd}{\Lmd+1}$ to both sides of the above identity, we arrive at
\begin{align*}
&
  U\frac{\Lmd-1}{\Lmd+1}
  \res\left(U\Lmd M^{p+1}\right)
=
  U(1-\Lmd)
  \res\left(\Lmd M^{p+1}\Lmd\right) \\
=&
  \res\left(U\Lmd M^{p+1} \Lmd\right)
 -\res\left(U\Lmd^2 M^{p+1}\right) \\
=&
  \res\left((L-\Lmd^2) M^{p+1} \Lmd\right)
 -\res\left(U\Lmd^2 M^{p+1}\right) \\
=&\,
  \res\left(M^p\Lmd\right)
 -\res\left(\Lmd^2 M^{p+1}(\Lmd+U)\right) \\
=&\,
  \res\left(M^p\Lmd\right)
 -\res\left(\Lmd^2 M^{p+1}L\Lmd^{-1}\right)
=
  \left(\Lmd^{-1}-\Lmd\right)
  \res\left(\Lmd M^p\right).
\end{align*}
Then by using \eqref{negative flow - res form}, \eqref{var der of H-p} and \eqref{second hamiltonian of neg flow} we obtain \eqref{negative flow - 2nd ham form}.
The proposition is proved.
\end{proof}

\begin{rmk}
The bihamiltonian structure \eqref{biham structure of q-KdV} associated with the Lax operator \eqref{zh-1-7} and the first negative flow of the $q$-deformed KdV hierarchy is given in \cite{Liuqp}.
\end{rmk}

\begin{rmk}\label{rmk:disp lim of p/m flows}
  Denote the symbol of the Lax operator $L$ by
\begin{align*}
  \lmd(z) &=z^2+Uz,
\end{align*}
then from \eqref{positive flow - res form} and \eqref{negative flow - res form} it follows that
the dispersionless limits of the flows \eqref{positive flow}--\eqref{negative flow} are given by
\begin{align*}
\lim_{\epsilon\to 0}\pfrac{U}{t^{1,p}}
&=
  -c_{1,p}U\p_x\left(\res_{z=0}\lmd^{p+\frac12}\cdot\frac{\td z}{z}\right) \notag \\
&=
  c_{1,p}\cdot \frac{(-1)^{p+1}}{2^{3p+1}}
  \frac{(2p+1)!!}{p!}U^{2p+1}U_x
=
  \frac{2}{p!}U^{2p+1}U_x,
\\
\lim_{\epsilon\to 0}
  \pfrac{U}{t^{0,-p}}
&=
  -2 c_{0,-p}\p_x\left(
    \res_{z=0}
      \frac{z}{\lmd^p}\cdot\frac{\td z}{z}
  \right)  \notag \\
&=\,
  (-1)^{p+1}c_{0,-p}
  \frac{2(2p-1)!}{(p-1)!^2}
  U^{-2p}U_x
=
(-1)^p\frac{(2p-1)!!}{2^p}U^{-2p}U_x,
\end{align*}
which coincide with the flows \eqref{PH-t1p}--\eqref{PH-t0-p} of the Principal Hierarchy of the 1-dimensional generalized Frobenius manifold $M$ if we identify $U$ with $v$.
\end{rmk}

\subsection{The logarithmic flows}\label{zh-1-8}
Let us proceed to construct the logarithm of the Lax operator \cite{extbitoda, exttoda} which is used in the definition of the flows \eqref{log flow}.
Note that the Lax operator $L$ given in \eqref{zh-1-7} can be represented in the form
\[
  L=(1+U\Lmd^{-1})\Lmd^2=(1+\Lmd U^{-1})(U\Lmd).
\]
Introduce the operators
\begin{align}
  X_1&=\log(1+U\Lmd^{-1})
       = \sum_{k=1}^{\infty}\frac{(-1)^{k-1}}{k}(U\Lmd^{-1})^k, \notag \\
  X_2&=\log(\Lmd^2)=2\epsilon\p_x, \notag\\
  Y_1&= \log(1+\Lmd U^{-1})=\sum_{k=1}^{\infty}\frac{(-1)^{k-1}}{k}(\Lmd U^{-1})^k, \notag \\
  Y_2&=\log(U\Lmd) = \epsilon\p_x+\frac{\epsilon\p_x}{\Lmd-1}(\log U), \label{240106-0630}
\end{align}
where the identity \eqref{240106-0630} can be derived from the integral form of the Baker--Campbell--Hausdorff (BCH) formula
(see, e.g. \cite{GTM222} or \cite{FVH}), then we have
\begin{equation}
  L=\rme^{X_1}\rme^{X_2}=\rme^{Y_1}\rme^{Y_2}.
\end{equation}
Now we define the operators $\log_-L$ and $\log_+L$ by using the BCH formula as follows:
\begin{align}
  \log_-L&=\log(\rme^{X_1}\rme^{X_2}) \notag\\
         &\phantom:= X_1+X_2+\frac12[X_1,X_2]+\frac{1}{12}
             \left(
               [X_1,[X_1,X_2]] + [X_2,[X_2,X_1]]
             \right)+\cdots , \label{log-L}\\
  \log_+L&=\log(\rme^{Y_1}\rme^{Y_2}) \notag\\
         &\phantom:= Y_1+Y_2+\frac12[Y_1,Y_2]+\frac{1}{12}
             \left(
               [Y_1,[Y_1,Y_2]] + [Y_2,[Y_2,Y_1]]
             \right)+\cdots. \label{log+L}
\end{align}
These operators can be represented in the form
\begin{align}
  \log_-L=2\epsilon\p_x+\sum_{k=1}^{\infty}a_k\Lmd^{-k},\quad
  \log_+L=\epsilon\p_x+\sum_{k=0}^{\infty}b_k\Lmd^k \label{coef ak bk}
\end{align}
for certain coefficients $a_k,b_k$.
We define the logarithm of $L$ by
  \begin{equation}
    \log L= 2\log_+L-\log_-L,
  \end{equation}
then it has the form
$\log L=\sum_{k\in\bbZ}f_k\Lmd^k$ for certain coefficients $f_k$.
For example, from \eqref{240106-0630} we obtain
\begin{equation}\label{240107-2305}
  f_0=2b_0=\frac{2\epsilon\p_x}{\Lmd-1}(\log U).
\end{equation}
Thus from \eqref{log flow} it follows that
\begin{align}
  \pfrac{U}{t^{0,0}}
&=\,
  \frac{c_{0,0}}{\epsilon}\res\left([f_0,U\Lmd]\Lmd^{-1}\right)
  =c_{0,0}U\frac{1-\Lmd}{\epsilon}f_0 \notag \\
&=\,
  -\frac12U\frac{1-\Lmd}{\epsilon}\frac{2\epsilon\p_x}{\Lmd-1}(\log U) \notag\\
&=\, U_x.
\end{align}
It can be verified that the positive, negative and logarithmic flows
\eqref{positive flow}--\eqref{log flow} are mutually commutative.

In general, the coefficients $a_k,b_k$ in \eqref{coef ak bk} for $k\neq 0$ are not easy to calculate,
so it is not a simple way to determine the flows $\pfrac{}{t^{0,p}}$ for $p>0$ from \eqref{log flow} by calculating $L^p\log L$ straightforwardly.
Instead, we can find the recursion relation between $\pfrac{}{t^{0,p}}$ and $\pfrac{}{t^{0,p+1}}$ and use it to determine the flows $\frac{\p}{\p t^{0,p}}$ recursively.

\begin{prop}
The logarithmic flows \eqref{log flow} satisfy the following bihamiltonian recursion relations
\begin{equation}\label{log recur relation}
  \mcalR\pfrac{U}{t^{0,p-1}}=p\pfrac{U}{t^{0,p}},\quad p\ge 1,
\end{equation}
where
\[\mcalR=\mcalP_2\mcalP_1^{-1}
=4U\frac{\Lmd-1}{\Lmd+1}U\frac{1}{\Lmd-\Lmd^{-1}}\]
is the recursion operator given by the bihamiltonian structure $(\mcalP_1,\mcalP_2)$ defined in \eqref{biham structure of q-KdV}.
\end{prop}

\begin{proof}
Denote
\begin{equation}\label{f_s^[k]}
  L^p\log L = \sum_{s\in\bbZ}f^{[p]}_s\Lmd^s,\quad p\ge 0,
\end{equation}
then from \eqref{log flow} we obtain
\begin{align}
  \epsilon\pfrac{U}{t^{0,p}}
&=\, c_{0,p}\res\left(\left[(L^p\log L)_+,L\right]\Lmd^{-1}\right)
 = c_{0,p}\res\left([f_0^{[p]}, U\Lmd]\Lmd^{-1}\right) \notag\\
&=\,
  c_{0,p}U(1-\Lmd)f_0^{[p]}. \label{240108-1049-1}
\end{align}
On the other hand, since $[L,\log L]=0$, the flows $\pfrac{}{t^{0,p}}$ can also be rewritten as
\begin{align}
\epsilon\pfrac{U}{t^{0,p}}
&=-c_{0,p}\res\left(\left[(L^p\log L)_-,L\right]\Lmd^{-1}\right)
 = -c_{0,p}\res\left(\left[f_{-1}^{[p]}\Lmd^{-1},\Lmd^2\right]\Lmd^{-1}\right) \notag \\
&=
  c_{0,p}(\Lmd-\Lmd^{-1})\left(f_{-1}^{[p]}\right)^+. \label{240108-1049-2}
\end{align}
Note that
\begin{align}
  f_0^{[p]}&=\res\left(L^p\log L\right) = \res\left(L^{p-1}\log L\cdot L\right) \notag \\
  &=\res\left((f_{-2}^{[p-1]}\Lmd^{-2}+f_{-1}^{[p-1]}\Lmd^{-1})(\Lmd^2+U\Lmd)\right) \notag \\
  &=\, f_{-2}^{[p-1]}+U^-f_{-1}^{[p-1]}.\label{241017-0945-1}
\end{align}
On the other hand, we also have
\begin{align}
  f_0^{[p]}&=\res\left(L\cdot L^{p-1}\log L\right)
  = \res\left(
      (\Lmd^2+U\Lmd)
      (f_{-2}^{[p-1]}\Lmd^{-2}+f_{-1}^{[p-1]}\Lmd^{-1})
    \right) \notag \\
  &=
    \left(f_{-2}^{[p-1]}\right)^{++}
   +U\left(f_{-1}^{[p-1]}\right)^+. \label{241017-0945-2}
\end{align}
Applying the operator $\Lmd^2$ on both sides of \eqref{241017-0945-1} and then subtracting \eqref{241017-0945-2} from it, we obtain
\[
  (\Lmd^2-1)f_0^{[p]}
 =
  (\Lmd-1)\left(U\left(f_{-1}^{[p-1]}\right)^+\right) ,
\]
therefore
\begin{equation}\label{240108-1301}
  (\Lmd-1)f_0^{[p]} =
\left(
  \frac{\Lmd-1}{\Lmd+1}\circ U\right)\left(
f_{-1}^{[p-1]}\right)^+.
\end{equation}
Thus from \eqref{240108-1049-1}--\eqref{240108-1049-2} and \eqref{240108-1301} we arrive at
\begin{align*}
  \epsilon\pfrac{U}{t^{0,p}}&=\,
c_{0,p}U(1-\Lmd)f_0^{[p]} =
  -c_{0,p}\left(U\circ\frac{\Lmd-1}{\Lmd+1}\circ U\right)\left(
  f_{-1}^{[p-1]}\right)^+ \\
&=\,
  -\frac{c_{0,p}}{c_{0,p-1}}
   \left(
     U\circ\frac{\Lmd-1}{\Lmd+1}\circ U\circ\frac{1}{\Lmd-\Lmd^{-1}}
   \right)
   \left(
     c_{0,p-1}(\Lmd-\Lmd^{-1})
     \left(f_{-1}^{[p-1]}\right)^+
   \right)\\
&=
  \frac {\epsilon}p\mcalR\pfrac V{t^{0,p-1}},\quad p\ge 1.
\end{align*}
The proposition is proved.
\end{proof}

\begin{thm}\label{thm:biham log flow}
The logarithmic flows \eqref{log flow} admit the following bihamiltonian formalism
\begin{align}
  \pfrac{U}{t^{0,p}}&=\mcalP_1\frac{\delta H_{0,p}}{\delta U}, \label{log 1st ham}\\
  p\pfrac{U}{t^{0,p}}&=\mcalP_2\frac{\delta H_{0,p-1}}{\delta U} \label{log 2nd ham},
\end{align}
here $p\geq 0$, and the Hamiltonian $H_{0,p}\,(p\geq -1)$ are defined by
\begin{align}\label{H-0p}
  H_{0,p}&=
    (-1)^{p+1}\frac{2^{2p}}{(p+1)!}\int\res\left(L^{p+1}\log L\right)\td x.
\end{align}
\end{thm}
We note that the Hamiltonian $H_{0,-1}$ defined in \eqref{H-0p} can be rewritten as
\begin{align*}
  H_{0,-1} = \frac14\int\res(\log L)\td x
=
  \frac 14\int\frac{2\epsilon\p_x}{\Lmd-1}(\log U)\td x
= \frac12\int\log U\td x,
\end{align*}
which coincides with \eqref{h0,-1 nega flow}. Indeed, it turns out that \eqref{log 2nd ham}
holds true for all $p\in\bbZ$.

\begin{proof}[Proof of Theorem \ref{thm:biham log flow}]
 By following the approach used in the proof of Lemma 3.5 of \cite{FVH} we obtain
  \[
    \delta(\log_+L-\epsilon\p_x) =
    \delta(\log_+L)=
    \frac{\ad_{\log_+L}}{1-\exp(-\ad_{\log_+L})}(M\delta L),
  \]
where $\log_+L$ is defined in \eqref{log+L}, $M=L^{-1}$ is defined in \eqref{zh-1-9}, \eqref{zh-1-10}, and $\ad_XY=[X,Y]$ for any operators $X,Y$.
Therefore for each $p\geq 0$ we have
  \begin{align*}
  &
    \delta\int\res\left(
      L^{p+1}(\log_+L-\epsilon\p_x)
    \right)\td x \\
  =&
    (p+1)\int\res\left(
      (\delta L\cdot L^p)(\log_+L-\epsilon\p_x)
    \right)\td x \\
  &
    +
    \int\res\left(
      L^{p+1}\cdot
      \frac{\ad_{\log_+L}}{1-\exp(-\ad_{\log_+L})}(M\delta L)
    \right)\td x \\
  =&
    (p+1)\int\res\left(
      \delta U\cdot\Lmd L^p(\log_+L-\epsilon\p_x)
    \right)\td x \\
  &
    -
    \int\res\left(
      M\delta L\cdot
      \frac{\ad_{\log_+L}}{1-\exp(-\ad_{\log_+L})}(L^{p+1})
    \right)\td x \\
  =&
    (p+1)\int\res\left(
      \delta U\cdot\Lmd L^p(\log_+L-\epsilon\p_x)
    \right)\td x
    -\int\res(L^p\delta L)\td x\\
  =&
  (p+1)\int\res\left(
      \delta U\cdot\Lmd L^p(\log_+L-\epsilon\p_x)
    \right)\td x ,
  \end{align*}
from which it follows that
\[
  \frac{\delta}{\delta U}
  \int\res\left(
      L^{p+1}(\log_+L-\epsilon\p_x)
    \right)\td x
=
  (p+1)\res\left(
      \Lmd L^p(\log_+L-\epsilon\p_x)
    \right).
\]
In a similar way we can derive the following formula
\[
  \frac{\delta}{\delta U}
  \int\res\left(
      L^{p+1}(\log_-L-2\epsilon\p_x)
    \right)\td x
=
  (p+1)\res\left(
      \Lmd L^p(\log_-L-2\epsilon\p_x)
    \right),
\]
and we have
\begin{align}
  \frac{\delta H_{0,p}}{\delta U}
&=
 \frac{(-1)^{p+1}2^{2p}}{(p+1)!}
  \frac{\delta}{\delta U}
  \int\res\left(
    L^{p+1}
      \left(
        2(\log_+L-\epsilon\p_x)
       -(\log_-L-2\epsilon\p_x)
      \right)
  \right)\td x \notag \\
&=
  \frac{(-1)^{p+1}2^{2p}}{p!}
  \res\left(\Lmd L^p\log L\right)
=
 2c_{0,p}\left(f_{-1}^{[p]}\right)^+,
\end{align}
here the coefficients $c_{0,p}$ and $f_s^{[p]}$ are defined in \eqref{zh-1-11} and \eqref{f_s^[k]} respectively.
Then by using \eqref{240108-1049-2} we obtain
\[
  \mcalP_1\frac{\delta H_{0,p}}{\delta U}
=
  \frac{1}{2\epsilon}(\Lmd-\Lmd^{-1})\cdot 2c_{0,p}\left(f_{-1}^{[p]}\right)^+
=
  \pfrac{U}{t^{0,p}},
\]
which leads to \eqref{log 1st ham} for all $p\geq 0$.
Thus from the recursion relations \eqref{log recur relation}
we arrive at \eqref{log 2nd ham} for $p\geq 1$.
It is easy to see that \eqref{log 2nd ham} also holds true when $p=0$.
The theorem is proved.
\end{proof}

\begin{rmk}
The dispersionless limit of the logarithmic flows can be obtained by the same way as we do for the positive and negative flows in Remark \ref{rmk:disp lim of p/m flows}.
Introduce the formal symbol
\[
  \log\lmd(z):=2\log V + 2\sum_{k=1}^{\infty}\frac{(-1)^{k-1}}{k}(\frac zU)^k
  -\sum_{k=1}^{\infty}\frac{(-1)^{k-1}}{k}(\frac Uz)^k,
\]
then from \eqref{240108-1049-1} it follows that
\begin{align*}
  \lim_{\epsilon\to 0}
  \pfrac U{t^{0,p}}
&=
  -c_{0,p}U\p_x\res_{z=0}
  \left(
    \lmd^p\log\lmd\cdot\frac{\td z}{z}
  \right) \\
&=
  (-1)^{p+1}c_{0,p}
  \left(
    \sum_{k=0}^{p}
      {p\choose k}\frac{(-1)^k}{p+k}
  \right)
  U\p_x(U^{2p}) \\
&=
   \frac{2^{2p-1}(p-1)!}{(2p-1)!}U^{2p}U_x,
\end{align*}
which coincide with the flows \eqref{PH-t0p} in the Principal Hierarchy of the generalized Frobenius manifold $M$ if we identify $U$ with $v$.
Therefore, the Principal hierarchy \eqref{PH-t1p}--\eqref{PH-t0p} is exactly the dispersionless extended $q$-deformed KdV hierarchy.
\end{rmk}

\subsection{The fractional Volterra hierarchy}

Let us recall the definition of the fractional Volterra hierarchy (FVH) and its basic properties, for details see \cite{FVH, Crelles}.
For any $\sfp,\sfq,\sfr\in\bbC$ which satisfy the local Calabi--Yau condition
\begin{equation}\label{CY condition}
  \frac1\sfp+\frac1\sfq+\frac1\sfr=0,
\end{equation}
we introduce the operators
\[
  \Lmd_1=\Lmd^{\frac1\sfq},\quad
  \Lmd_2=\Lmd^{\frac1\sfp},\quad
  \Lmd_3=\Lmd^{\frac1\sfp+\frac1\sfq}
\]
and the Lax operator
\begin{equation}\label{general FVH Lax}
  \mcalL=\Lmd_2+\rme^W\Lmd_1^{-1},
\end{equation}
where $W$ is the unknown function. The fractional Volterra hierarchy of $(\sfp,\sfq,\sfr)$-type
consists of the following flows:
\begin{align}
  \epsilon\pfrac{\mcalL}{T_s}
 &=
    \left[
      \left(\mcalL^{-\frac s\sfr}\right)_\oplus , \mcalL
    \right], \quad s\in\ell_1,\label{general FVH-1} \\
 \epsilon\pfrac{\mcalL}{T_s}
 &=  -\left[
      \left(\mcalL^{-\frac s\sfr}\right)_\ominus , \mcalL
    \right], \quad s\in\ell_2,\label{general FVH-2}
 \end{align}
here the index sets are given by
\begin{equation}\label{index set l1,l2}
  \ell_1=\Bigset{k\sfp}{k\in\bbZ_{\geq 0}},\quad
  \ell_2=\Bigset{k\sfq}{k\in\bbZ_{\geq 0}}.
\end{equation}
Note that for $s=k\sfp\in\ell_1$, the operator $\mcalL^{-\frac s\sfr}$ is of the form
\[
  \mcalL^{-\frac s\sfr}=\sum_{l\geq 0} a_l\Lmd_3^{k-l}
\]
for certain coefficients $a_l$, which is a Laurent series of $\Lmd_3^{-1}$, and
its positive part (with respect to $\Lmd_3$) is defined as $\left(\mcalL^{-\frac s\sfr}\right)_\oplus:=\sum_{l=0}^k a_l\Lmd_3^{k-l}$.
Similarly, when $s=k\sfq\in\ell_2$ we know that
\[
  \mcalL^{-\frac s\sfr} = \sum_{l\geq 0} b_l\Lmd_3^{-k+l}
\]
is a Laurent series of $\Lmd_3$,
and its negative part (with respect to $\Lmd_3$) is defined as
$\left(\mcalL^{-\frac s\sfr}\right)_\ominus:=\sum_{l=0}^{k-1}b_l\Lmd_3^{-k+l}$.

An important case of the fractional Volterra hierarchy is when $(\sfp,\sfq,\sfr)=(-\frac12,1,1)$,
which is closely related with the extended $q$-deformed KdV hierarchy and the one-dimensional generalized Frobenius manifold $M$ with potential \eqref{F=v^4/12}.

\begin{ex}\label{ex: -1/2,1,1}
Let $(\sfp,\sfq,\sfr)=(-\frac12,1,1)$, then we have
  \[
    \Lmd_1=\Lmd, \quad \Lmd_2=\Lmd^{-2},\quad \Lmd_3=\Lmd^{-1},
  \]
  and the corresponding Lax operator
  \begin{equation}\label{special FVH L}
    \mcalL=\Lmd^{-2}+\rme^W\Lmd^{-1}.
  \end{equation}
The non-trivial flows of the fractional Volterra hierarchy are given by
\begin{align}
  \epsilon\pfrac{\mcalL}{T_{-p-\frac12}}&=
    \left[
      \left(\mcalL^{p+\frac12}\right)_\oplus , \mcalL
    \right],\quad p\geq 0, \label{240118-1719-1}\\
  \epsilon\pfrac{\mcalL}{T_p}&=
    -\left[
      \left(\mcalL^{-p}\right)_\ominus , \mcalL
    \right],\quad p\geq 1. \label{240118-1719-2}
\end{align}
\end{ex}

Now we identify the unknown functions and the time variables as follows:
\begin{align}
W&=\log U,\\
 T_{-p-\frac12}&=-c_{1,p}t^{1,p},\quad p\geq 0, \label{Tt-1}\\
 T_{p} &= c_{0,-p}t^{0,-p},\quad p\geq 1, \label{Tt-2}
\end{align}
where the constants $c_{1,p}$ and $c_{0,-p}$ are defined in \eqref{zh-1-11}, \eqref{zh-1-12}.
In particular,
\begin{equation}\label{t^1,0 and T-1/2}
  t^{1,0}=\frac14T_{-\frac12}.
\end{equation}

The relation between the $(-\frac12,1,1)$-type  fractional Volterra hierarchy and the extended $q$-deformed KdV hierarchy is given in the following lemma.
\begin{lem}\label{lem: -1/2,1,1 and qKdV}
Suppose $W$ satisfies the $(-\frac12,1,1)$-type fractional Volterra hierarchy \eqref{240118-1719-1}, \eqref{240118-1719-2},
then $U=\rme^W$ satisfies the positive and negative flows \eqref{positive flow}, \eqref{negative flow} of the extended $q$-deformed hierarchy under the identification \eqref{Tt-1}, \eqref{Tt-2}.
\end{lem}
\begin{proof}
The Lax operator \eqref{zh-1-7} of the extended $q$-deformed KdV hierarchy and that of the $(-\frac12,1,1)$-type fractional Volterra hierarchy \eqref{special FVH L} satisfy the following relation
  \[
    L=\mcalL|_{\epsilon\mapsto -\epsilon}.
  \]
  Note that $\Lmd_3=\Lmd^{-1}$ when $(\sfp,\sfq,\sfr)=(-\frac12,1,1)$,
 by  checking the definition of the positive parts $(\cdot)_\oplus$ and $(\cdot)_+$ carefully, we obtain
  \begin{align*}
    \pfrac{L}{t^{1,p}}
  &=
    -c_{1,p}\pfrac{L}{T_{-p-\frac12}}
   =-c_{1,p}\left.
     \pfrac{\mcalL}{T_{-p-\frac12}}
     \right|_{\epsilon\mapsto -\epsilon} \\
  &=
    -c_{1,p}
     \left.
     \left(
       \frac{1}{\epsilon}
       \left[
         \left(\mcalL^{p+\frac12}\right)_\oplus , \mcalL
       \right]
     \right)
     \right|_{\epsilon\mapsto -\epsilon}
  =
    \frac{c_{1,p}}{\epsilon}
    \left[
      \left(L^{p+\frac12}\right)_+,L
    \right].
  \end{align*}
Similarly, $U=\rme^W$ also satisfies the negative flows \eqref{negative flow} of the extended $q$-deformed KdV hierarchy.
The lemma is proved.
\end{proof}

\subsection{The quasi-triviality of the FVH}

For each $\sfp,\sfq,\sfr\in\bbC$ satisfying the local Calabi--Yau condition \eqref{CY condition},
it is shown in \cite{Crelles} that the associated fractional Volterra hierarchy \eqref{general FVH-1}, \eqref{general FVH-2} admits a tau structure, from which it follows that for any solution $W(x,T;\epsilon)$ of the fractional Volterra hierarchy there is a tau function $\tau(x,T;\epsilon)$ such that
\begin{align}
&W=
\left(\Lmd^{-\frac{1}{2\sfr}}-\Lmd^{\frac{1}{2\sfr}}\right)
\left(\Lmd^{\frac{1}{2\sfq}}-\Lmd^{-\frac{1}{2\sfq}}\right)
\log\tau,\\
& \epsilon\left(\Lmd^{-\frac 1\sfr}-1\right)\Lmd^{-\frac{1}{2\sfp}}
  \pfrac{\log\tau}
  {T_s}
=
  \res \mcalL^{-\frac s\sfr},\quad s\in\ell_1\cup\ell_2.\label{e^W tau relation}
\end{align}
It is shown in \cite{Crelles} that the fractional Volterra hierarchy possesses a topological tau function $\tau_{\mathrm{top}}$ which has the genus expansion
\[\log\tau_{\mathrm{top}}
=
  \frac1{\epsilon^2}\mathcal{H}_{0,\mathrm{top}}(x, T)
 +\sum_{g\ge 1} \epsilon^{2g-2} \mcalH_g\left(w_{\mathrm{top}}, w_{\mathrm{top}}',\dots,w_{\mathrm{top}}^{(3g-2)}\right),\]
where
\[w_{\mathrm{top}}(x,T)=\frac{\sfp+\sfq}{\sfp \sfq^2}\p_x^2 \mathcal{H}_{0,\mathrm{top}}(x, T)\]
is a certain solution of the dispersionless limit of the fractional Volterra hierarchy, and the function
\[\Delta \mcalH(\epsilon)
=
\sum_{g\ge 1} \epsilon^{2g-2} \mcalH_g\left(w, w_x,\dots,w^{(3g-2)}\right)\]
satisfies the loop equation
\begin{align}
  &\,
    \sum_{s\geq 0}
      \left(
        \p_x^s\Theta
       +\sum_{k=1}^{s}
          {s\choose k}P_{k-1,s-k+1}
      \right)\pfrac{\Delta \mcalH}{w^{(s)}} \notag\\
  =&\,
    \frac{1}{\sfp}\left(
      \frac{\Theta^2}{16} - \left(\frac{1}{16}-\frac{\sigma_1}{24}\right)\Theta
    \right)
   +\frac{\epsilon^2}{\sigma_2}\sum_{s\geq 0}
      \p_x^{s+2}
      \left(
        \frac{\Theta^2}{16}
       -\left(\frac{1}{16}-\frac{\sigma_1}{24}\right)\Theta
      \right) \pfrac{\Delta \mcalH}{w^{(s)}} \notag\\
   &\quad
   +\frac{\sfp \epsilon^2}{2\sigma_2}
    \sum_{k,\ell\geq 0}P_{k+1,\ell+1}
      \left(
        \frac{\p^2\Delta \mcalH}{\p w^{(k)}\p w^{(\ell)}}
       +\pfrac{\Delta \mcalH}{w^{(k)}}
        \pfrac{\Delta \mcalH}{w^{(\ell)}}
      \right). \label{Hodge loop eqn}
  \end{align}
Here in the loop equation $\sigma_1, \sigma_2, \Theta$ are defined by
\[
  \sigma_1=-(\sfp+\sfq+\sfr),\quad
  \sigma_2=-\sfp\sfq\sfr,\quad
   \Theta=\frac{1}{1-\lmd\rme^{\frac{w}{\sfp}}},
\]
$P_{k,\ell}$ are certain polynomials in $\Theta,\sigma_1,\sigma_2$ and the jet variables $w_x, w_{xx},...$
whose definitions are given in Sect.\,4 of \cite{JDG},
and the above equation is required to hold true identically with respect to the parameter $\lmd$.

By using the transcendency of $w_{\mathrm{top}}(x, T)$, it is shown in \cite{Crelles} that for any solution $w(x, T)$ of the dispersionless fractional Volterra hierarchy, the function
\[W=\left(\Lmd^{-\frac{1}{2\sfr}}-\Lmd^{\frac{1}{2\sfr}}\right)
\left(\Lmd^{\frac{1}{2\sfq}}-\Lmd^{-\frac{1}{2\sfq}}\right)
\left(\frac1{\epsilon^2}\mathcal{H}_{0}(x, T)+
\Delta\mcalH\left(
\epsilon\right)\right)\]
is a solution to the fractional Volterra hierarchy \eqref{general FVH-1}, \eqref{general FVH-2}, here $\rme^{\mathcal{H}_0(x, T)}$ is the tau function
of the solution $w(x,T)$ of the dispersionless fractional Volterra hierarchy which is related to $w(x, T)$ by
\[w(x,T)=\frac{\sfp+\sfq}{\sfp \sfq^2}\p_x^2 \mathcal{H}_0(x, T).\]
Moreover, the tau function $\tau$ of the solution $W$ to the fractional Volterra hierarchy is given by
\[\log\tau=\frac1{\epsilon^2}\mathcal{H}_{0}(x, T)+
\Delta\mcalH\left(
\epsilon\right).\]

Now let us consider Example \ref{ex: -1/2,1,1}, i.e., we assume that
\[(\sfp, \sfq, \sfr)=\left(-\frac12, 1, 1\right),\]
then in the above loop equation $P_{i,j}$ take the following explicit forms:
\[P_{k,\ell}=\left(\p_x^k\sqrt{\Theta}\right)\left(\p_x^\ell\sqrt{\Theta}\right),\quad k, \ell\ge 0, \]
and the loop equation \eqref{Hodge loop eqn} for the case $(\sfp,\sfq,\sfr)=\left(-\frac12,1,1\right)$ is given by
\begin{align}
&
  \sum_{s\geq 0}
    \left(
      \p_x^s\Theta
     +\sum_{k=1}^{s}{s\choose k}
      \left(\p_x^{k-1}\sqrt{\Theta}\right)
      \left(
        \p_x^{s-k+1}\sqrt{\Theta}
      \right)
    \right)
    \pfrac{\Delta\mcalH}{w^{(s)}} \notag\\
=&
  -\frac12\epsilon^2
  \sum_{k,\ell\geq 0}
    \left(\p_x^{k+1}\sqrt{\Theta}\right)
    \left(\p_x^{\ell+1}\sqrt{\Theta}\right)
    \left(
      \frac{\p^2\Delta\mcalH}{\p w^{(k)}\p w^{(\ell)}}
     +\pfrac{\Delta\mcalH}{w^{(k)}}
      \pfrac{\Delta\mcalH}{w^{(\ell)}}
    \right) \notag\\
&
  +\epsilon^2\sum_{s\geq 0}
    \p_x^{s+2}
    \left(\frac{\Theta^2}{8}-\frac\Theta 4\right)
    \pfrac{\Delta\mcalH}{w^{(s)}}
   -\left(\frac{\Theta^2}{8}-\frac\Theta 4\right),  \label{-1/2,1,1 loop eqn}
\end{align}
here $\Theta = \frac{\rme^{2w}}{\rme^{2w}-\lmd}$.
The first three
lower-genus components of its solution are given by
\begin{align}
  \mcalH_1 &=\frac1{24}\log w_x+\frac 18w, \label{H1-sol}\\
  \mcalH_2 &= -\frac{w^{(4)}}{1152w_x^2}
    +\frac{7w^{(3)}w_{xx}}{1920w_x^3}
    -\frac{w^{(3)}}{160w_x}
    -\frac{w_{xx}^3}{360w_x^4}
    +\frac{11w_{xx}^2}{1920w_x^2}
    -\frac{7w_{xx}}{640}
    -\frac{w_x^2}{1440}, \label{H2-sol}\\
  \mcalH_3 &=
     \frac{w^{(7)}}{82944 w_x^3}
    -\frac{7 w^{(6)} w_{xx}}{46080 w_x^4}
    +\frac{7 w^{(6)}}{46080 w_x^2}
    -\frac{53 w^{(3)} w^{(5)}}{161280 w_x^4}
    +\frac{353 w^{(5)} w_{xx}^2}{322560 w_x^5} \notag\\
&\quad
    -\frac{383 w^{(5)} w_{xx}}{322560 w_x^3}
    +\frac{41 w^{(5)}}{64512 w_x}
    -\frac{103 (w^{(4)})^2}{483840 w_x^4}
    +\frac{1273 w^{(4)}w^{(3)}  w_{xx}}{322560 w_x^5}
    -\frac{689 w^{(4)}w^{(3)} }{322560 w_x^3} \notag\\
&\quad
    -\frac{83 w^{(4)} w_{xx}^3}{15120 w_x^6}
    +\frac{185 w^{(4)} w_{xx}^2}{32256 w_x^4}
    -\frac{373 w^{(4)} w_{xx}}{161280 w_x^2}
    +\frac{185 w^{(4)}}{193536}
    +\frac{59 (w^{(3)})^3}{64512 w_x^5} \notag\\
&\quad
    -\frac{83 (w^{(3)})^2 w_{xx}^2}{7168 w_x^6}
    +\frac{869 (w^{(3)})^2 w_{xx}}{107520 w_x^4}
    -\frac{61 (w^{(3)})^2}{35840 w_x^2}
    +\frac{59 w^{(3)} w_{xx}^4}{3024 w_x^7}
    -\frac{9343 w^{(3)} w_{xx}^3}{483840 w_x^5}  \notag\\
&\quad
    +\frac{151 w^{(3)} w_{xx}^2}{23040 w_x^3}
    -\frac{19 w^{(3)} w_{xx}}{120960 w_x}
    +\frac{41 w^{(3)} w_x}{120960}
    -\frac{5 w_{xx}^6}{648 w_x^8}
    +\frac{131 w_{xx}^5}{15120 w_x^6}
    -\frac{57 w_{xx}^4}{17920 w_x^4}  \notag\\
&\quad
    +\frac{w_{xx}^3}{9072 w_x^2}
    +\frac{31 w_{xx}^2}{120960}
    -\frac{w_x^4}{90720}. \label{H3-sol}
\end{align}

\subsection{The loop equation of the one dimensional GFM}

 For the one dimensional generalized Frobenius manifold \eqref{F=v^4/12}, its period can be chosen as
 \[
   p(v;\lmd) = \log\left(v+\sqrt{v^2-\lmd}\right),
 \]
 with the associated Gram matrix $(G^{\afa\beta})=1$ and the star product
 \[
   \pfrac{p}{\lmd}*\pfrac{p}{\lmd}
  =\frac{1}{8\lmd^2}-\frac{1}{8(v^2-\lmd)^2},
 \]
 therefore the loop equation of this generalized Frobenius manifold can be written in the form
 \begin{align}
&
   \sum_{s\geq 0}
   \pfrac{\Delta\mcalF}{v^{(s)}}
   \p_x^s\frac{1}{2v(v^2-\lmd)}
 +
  \sum_{s\geq 1}
   \pfrac{\Delta\mcalF}{v^{(s)}}
   s\p_x^s
   \left[
     \frac{1}{2\lmd}
     \left(
       \frac{1}{\sqrt{v^2-\lmd}}-\frac 1v
     \right)
   \right]  \notag \\
&
  +\sum_{s\geq 1}
   \pfrac{\Delta\mcalF}{v^{(s)}}
   \sum_{k=1}^{s}
   {s\choose k}
   \left[
     \p_x^{k-1}
     \frac{1}{2\lmd}
     \left(
       \frac{v}{\sqrt{v^2-\lmd}}-1
     \right)
   \right]
   \left(
     \p_x^{s+1-k}
     \frac{1}{\sqrt{v^2-\lmd}}
   \right)  \notag \\
=&
  \frac12\veps^2
  \sum_{k,\ell\geq 0}
  \left(
    \pfrac{\Delta\mcalF}{v^{(k)}}
    \pfrac{\Delta\mcalF}{v^{(\ell)}}
   +\frac{\p^2\Delta\mcalF}{\p v^{(k)}\p v^{(\ell)}}
  \right)
  \left(
    \p_x^{k+1}
    \frac{1}{\sqrt{v^2-\lmd}}
  \right)
  \left(
    \p_x^{\ell+1}
    \frac{1}{\sqrt{v^2-\lmd}}
  \right)  \notag \\
&+
   \frac12\veps^2
     \sum_{k\geq 0}
     \pfrac{\Delta\mcalF}{v^{(k)}}
     \p_x^{k+1}
     \frac{v^2v_x}{(v^2-\lmd)^3}
 -
   \frac{1}{16}\frac{1}{(v^2-\lmd)^2},\label{Loop eqn for v^4/12}
\end{align}
see also in \cite{GFM},
where the unknown function $\Delta\mcalF$ has the genus expansion
\[\Delta\mcalF=\Delta\mcalF(\veps)=\sum_{g\geq 1}\veps^{2g-2}\mcalF_g(v,v_x,...,v^{(3g-2)}),\]
and the above loop equation is required to hold true identically with respect to the parameter $\lmd$.
The first three lower-genus components of the solution to this loop equation are given by
\begin{align}
\mcalF_1 &=\, \frac{1}{24}\log v_x+\frac{1}{12}\log v,  \label{F1 in v^4/12}\\
\mcalF_2 &=\,
 \frac{ v^{(4)}v}{576 v_x^2}
-\frac{7v^{(3)}v_{xx}v}{960 v_x^3}
+\frac{37 v^{(3)}}{2880 v_x}
+\frac{v_{xx}^3 v}{180 v_x^4}
-\frac{11 v_{xx}^2}{960 v_x^2}
+\frac{v_{xx}}{120 v}
-\frac{v_x^2}{120 v^2},  \label{F2 in v^4/12}
\\
\mcalF_3 &=\,
  \frac{    v^{(7)}v^2 }{20736 v_x^3}
 -\frac{7   v^{(6)}v_{xx}v^2}{11520 v_x^4}
 +\frac{91  v^{(6)}v }{103680 v_x^2}
 -\frac{53  v^{(5)}v^{(3)}v^2}{40320 v_x^4}
 +\frac{353 v^{(5)}v_{xx}^2 v^2}{80640 v_x^5} \notag \\
&\quad
 -\frac{419 v^{(5)}v_{xx}v }{60480 v_x^3}
 +\frac{913 v^{(5)}}{241920 v}
 -\frac{103 (v^{(4)})^2 v^2 }{120960 v_x^4}
 +\frac{1273  v^{(4)} v^{(3)}v_{xx}v^2 }{80640 v_x^5}\notag\\
 &\quad
 -\frac{9169 v^{(4)} v^{(3)} v }{725760 v_x^3}
 -\frac{83  v^{(4)} v_{xx}^3 v^2}{3780 v_x^6}
 +\frac{545 v^{(4)} v_{xx}^2 v}{16128 v_x^4}
 -\frac{3727v^{(4)}v_{xx}}{241920 v_x^2}
 +\frac{v^{(4)}}{1512 v}\notag\\
 &\quad
 +\frac{59 (v^{(3)})^3 v^2 }{16128 v_x^5}
 -\frac{83  (v^{(3)})^2 v_{xx}^2 v^2 }{1792 v_x^6}
 +\frac{97  (v^{(3)})^2v_{xx}v}{2016 v_x^4}
 -\frac{1669 (v^{(3)})^2}{145152 v_x^2}\notag\\
 &\quad
 +\frac{59 v^{(3)}  v_{xx}^4 v^2}{756 v_x^7}
 -\frac{5555  v^{(3)}v_{xx}^3v  }{48384 v_x^5}
 +\frac{325  v^{(3)} v_{xx}^2}{6912 v_x^3}
 -\frac{v^{(3)}v_x}{378 v^2}
 -\frac{5 v_{xx}^6v^2 }{162 v_x^8}\notag\\
 &\quad
 +\frac{13 v_{xx}^5v }{252 v_x^6}
 -\frac{193 v_{xx}^4}{8064 v_x^4}
 -\frac{v_{xx}^2}{504 v^2}
 +\frac{v_{xx}v_x^2 }{126 v^3}
 -\frac{v_x^4}{252 v^4}
 \label{F3 in v^4/12}.
\end{align}
Note that under the change of variable $w=\log v$, we have
\[
  \mcalF_1 = \mcalH_1,\quad
  \mcalF_2 = -2\mcalH_2,\quad
  \mcalF_3 = 4\mcalH_3,
\]
where $\mcalH_1,\mcalH_2,\mcalH_3$ are the first three components of the solution of the loop equation \eqref{-1/2,1,1 loop eqn},
which are given in \eqref{H1-sol}--\eqref{H3-sol}. In general, we have

\begin{prop}
Under the change of variable $w=\log v$ and the identification $\epsilon=\sqrt{2}\,\rmi\veps$, the following identity holds true
\begin{equation}\label{main lemma}
  \Delta\mcalF = \Delta\mcalH + \mathrm{const.},
\end{equation}
here
\[\Delta\mcalF=\Delta\mcalF(\veps)=\sum_{g\geq 1}\veps^{2g-2}\mcalF_g,\quad
\Delta\mcalH=\Delta\mcalH(\epsilon)=\sum_{g\geq 1}\epsilon^{2g-2}\mcalH_g\]
are the unique (up to adding a constant) solutions to the loop equations \eqref{Loop eqn for v^4/12} and \eqref{-1/2,1,1 loop eqn}
respectively.
\end{prop}

\begin{proof}
  Applying the formula \eqref{change jet variables} to the relation $w=\log v$, we obtain
  \begin{equation}
    \pp{w^{(s)}}=
    \sum_{t\geq s}{t\choose s}v^{(t-s)}\pp{v^{(t)}}.
  \end{equation}
Then under the above replacement and the identification $\epsilon = \sqrt{2}\,\rmi\veps$,
the loop equation \eqref{-1/2,1,1 loop eqn} for $\Delta\mcalH = \Delta\mcalH(\epsilon)$ can be rewritten in the $v$-coordinate as
\begin{align}
&
  \sum_{s\geq 0}
    \pfrac{\Delta\mcalH}{v^{(s)}}\p_x^s\frac{v^3}{v^2-\lmd}
 +
  \sum_{s\geq 1}
    \pfrac{\Delta\mcalH}{v^{(s)}}
    \sum_{k=1}^{s}{s\choose k}
    \left(\p_x^{k-1}\frac{v}{\sqrt{v^2-\lmd}}\right)
    \left(\p_x^{s+1-k}\frac{\lmd}{\sqrt{v^2-\lmd}}\right)  \notag\\
=&
  \veps^2
  \sum_{k,\ell\geq 0}
    \left(\p_x^{k+1}\frac{\lmd}{\sqrt{v^2-\lmd}}\right)
    \left(\p_x^{\ell+1}\frac{\lmd}{\sqrt{v^2-\lmd}}\right)
    \left(
      \frac{\p^2\Delta\mcalH}{\p v^{(k)}\p v^{(\ell)}}
     +\pfrac{\Delta\mcalH}{v^{(k)}}
      \pfrac{\Delta\mcalH}{v^{(\ell)}}
    \right)  \notag\\
&
  +\veps^2\sum_{k\geq 0}
    \left(
      \p_x^{k+1}
        \frac{\lmd^2v^2v_x}{(v^2-\lmd)^3}
    \right)\pfrac{\Delta\mcalH}{v^{(k)}}
  -\frac{2\lmd v^2-v^4}{8(v^2-\lmd)^2}.   \label{231013-1738}
\end{align}
On the other hand,
for the solution $\Delta\mcalF=\sum_{g\geq 1}\veps^{2g-2}\mcalF_g$ to \eqref{Loop eqn for v^4/12},
by multiplying both sides of \eqref{Loop eqn for v^4/12} by $2\lmd^2$
we obtain
\begin{align}
&
  \sum_{s\geq 0}\pfrac{\Delta\mcalF}{v^{(s)}}
  \p_x^s\frac{v^3}{v^2-\lmd}
 -\lmd\sum_{s\geq 0}\pfrac{\Delta\mcalF}{v^{(s)}}(s+1)\left(\p_x^s\frac 1v\right)
 -\sum_{s\geq 0}v^{(s)}\pfrac{\Delta\mcalF}{v^{(s)}} \notag \\
&
 +\sum_{s\geq 1}
   \pfrac{\Delta\mcalF}{v^{(s)}}
   \sum_{k=1}^{s}{s\choose k}
     \left(\p_x^{k-1}\frac{v}{\sqrt{v^2-\lmd}}\right)
     \left(\p_x^{s+1-k}\frac{\lmd}{\sqrt{v^2-\lmd}}\right)  \notag \\
=&
  \veps^2
  \sum_{k,\ell\geq 0}
    \left(\p_x^{k+1}\frac{\lmd}{\sqrt{v^2-\lmd}}\right)
    \left(\p_x^{\ell+1}\frac{\lmd}{\sqrt{v^2-\lmd}}\right)
    \left(
      \frac{\p^2\Delta\mcalF}{\p v^{(k)}\p v^{(\ell)}}
     +\pfrac{\Delta\mcalF}{v^{(k)}}
      \pfrac{\Delta\mcalF}{v^{(\ell)}}
    \right)  \notag\\
&
  +\veps^2\sum_{k\geq 0}
    \left(
      \p_x^{k+1}
        \frac{\lmd^2v^2v_x}{(v^2-\lmd)^3}
    \right)\pfrac{\Delta\mcalF}{v^{(k)}}
  -\frac{\lmd^2}{8(v^2-\lmd)^2}. \label{231014-1638}
\end{align}
Comparing the coefficients of $\frac{1}{\lmd}$ and $\frac{1}{\lmd^2}$ at $\lmd=\infty$ in the loop equation \eqref{Loop eqn for v^4/12},
we obtain
\begin{align}
  &\sum_{s\geq 0}(s+1)\left(\p_x^s\frac{1}{v}\right)\pfrac{\Delta\mcalF}{v^{(s)}}= 0, \label{231013-1729-1}\\
  &\sum_{s\geq 0}v^{(s)}\pfrac{\Delta\mcalF}{v^{(s)}}=\frac 18 , \label{231013-1729-2}
\end{align}
which are equivalent to the linearization conditions \eqref{linearization condition} for $m=-1, 0$ respectively.
From \eqref{231014-1638}--\eqref{231013-1729-2} and \eqref{231013-1738}, it is clear that the solution $\Delta\mcalF$ to \eqref{Loop eqn for v^4/12} is also a solution to \eqref{-1/2,1,1 loop eqn}. The proposition is proved.
\end{proof}

\begin{thm}
  Suppose $v=v(t)$ is a solution to the Principal Hierarchy \eqref{PH-t1p}--\eqref{PH-t0p}
  of the one dimensional GFM \eqref{F=v^4/12}, then
  \begin{equation}
    U=
      \frac{\Lmd-\Lmd^{-1}}{2\epsilon\p_x}
      \left(v+\veps^2\frac{\p^2\Delta\mcalF}{\p x \p t^{1,0}}\right)
  \end{equation}
  satisfies the extended $q$-deformed KdV hierarchy \eqref{positive flow}--\eqref{log flow},
  where $\epsilon = \sqrt{2}\,\rmi\veps$, $\Lmd=\rme^{\epsilon\p_x}$, and $\Delta\mcalF=\Delta\mcalF(\veps)$ is the solution to the loop equation \eqref{Loop eqn for v^4/12}.
\end{thm}

\begin{proof}
Applying formula \eqref{e^W tau relation} to the case $(\sfp,\sfq,\sfr)=\left(-\frac12,1,1\right)$ and $s=-\frac12$, we obtain
\[
  \epsilon(\Lmd^{-1}-1)\Lmd
  \pfrac{\log\tau}{T_{-\frac12}}=\res\mcalL^{\frac12}=\frac{1}{1+\Lmd^{-1}}\rme^W,
\]
which can be rewritten as
\begin{equation}\label{240130-2219}
  \epsilon(\Lmd^{-1}-\Lmd)
  \pp{T_{-\frac12}}
  \left(
    \frac{1}{\epsilon^2}\mcalH_0+\Delta\mcalH
  \right) = \rme^W,
\end{equation}
where $\Delta\mcalH=\Delta\mcalH(\epsilon)$ is the solution to the loop equation \eqref{-1/2,1,1 loop eqn}.
The dispersionless limit of the above formula is
\[
  v=\rme^w=-2\frac{\p^2\mcalH_0}{\p x\p T_{-\frac12}}.
\]
Then by using \eqref{t^1,0 and T-1/2}, \eqref{main lemma}, \eqref{240130-2219} and the identification $\epsilon=\sqrt{2}\,\rmi\veps$, we obtain
  \begin{align*}
  U&=
    \frac{\Lmd-\Lmd^{-1}}{2\epsilon\p_x}
    \left(
      v+\veps^2\frac{\p^2\Delta\mcalF}{\p x\p t^{1,0}}
    \right) \\
  &=
    \frac{\Lmd-\Lmd^{-1}}{2\epsilon\p_x}
    \left(
      -2\frac{\p^2\mcalH_0}{\p x\p T_{-\frac12}}
      +\left(\frac{\epsilon}{\sqrt{2}\,\rmi}\right)^2
      \cdot 4\frac{\p^2\Delta\mcalH}{\p x\p T_{-\frac12}}
    \right) \\
  &=
    \epsilon(\Lmd^{-1}-\Lmd)
    \pp{T_{-\frac12}}
    \left(
      \frac{1}{\epsilon^2}\mcalH_0
     +\Delta\mcalH
    \right)\\
  &=
    \rme^W.
  \end{align*}
  Since $W$ satisfies the fractional Volterra hierarchy \eqref{240118-1719-1}, \eqref{240118-1719-2} of $(-\frac12,1,1)$-type,
we know from Lemma \ref{lem: -1/2,1,1 and qKdV} that $U=\rme^W$ satisfies the positive and negative flows of the $q$-deformed KdV hierarchy.

Finally, we note that the quasi-miura transformation
\[v\mapsto U=v+\veps A^{[1]}+\veps^2 A^{[2]}+\cdots\]
transforms the $\pp{t^{0,p}}$-flows $(p\geq 0)$ of the Principal Hierarchy to some bihamiltonian flows which are deformations of the dispersionless logarithmic flows \eqref{log flow} of the extended $q$-deformed KdV hierarchy. From the triviality of the first bihamiltonian cohomology of the dispersionless limit of the bihamiltonian structure \eqref{biham structure of q-KdV}, we know that such deformations of bihamiltonian flows are unique \cite{DLZ06, Liusq}, thus the function $U$ satisfies the logarithmic flows of the extended $q$-deformed KdV hierarchy. The theorem is proved.
\end{proof}

\subsection{The Volterra hierarchy and the $q$-deformed KdV hierarchy}

Another important case of the fractional Volterra hierarchy \eqref{general FVH-1}, \eqref{general FVH-2} is when $(\sfp,\sfq,\sfr)=(1,1,-\frac12)$.
In this case the shift operator $\Lmd_3=\Lmd^2$, and the Lax operator
$\mcalL$, which we now denote by $\tilde{\mcalL}$, has the form
\[
  \tilde{\mcalL}=\Lmd+\rme^{\tilde{W}}\Lmd^{-1},
\]
here we denote the unknown function $W$ by $\tilde{W}$. The nontrivial flows of the integrable hierarchy are given by
\begin{equation} \label{Volterra}
  \epsilon\pfrac{\tilde{\mcalL}}{\tilde{T}_p}
 = \left[
   \left(
     \tilde{\mcalL}^{2p}
   \right)_+ , \tilde{\mcalL}
 \right],\quad p\geq 1.
\end{equation}
This integrable hierarchy is called the Volterra hierarchy (also known as the discrete KdV hierarchy), which
has the following bihamiltonian structure \cite{Adler, FT, DLYZ1}:
\begin{align*}
  \tilde{\mcalP}_1&=
    \frac{1}{2\epsilon}
    \left(
      (\Lmd+1)\rme^{\tilde{W}}(\Lmd+1)
     -(1+\Lmd^{-1})\rme^{\tilde{W}}(1+\Lmd^{-1})
    \right), \\
  \tilde{\mcalP}_2&=
    \frac{2}{\epsilon}(\Lmd-\Lmd^{-1}).
\end{align*}
By a straightforward calculation one can verify that
the bihamiltonian structures $(\tilde{\mcalP}_1,\tilde{\mcalP}_2)$ and $(\mcalP_1,\mcalP_2)$ (see its definition given in \eqref{biham structure of q-KdV})
are related by the change of the unknown function
\begin{equation}\label{240119-2005}
  \tilde{W}=-(\Lmd^{\frac12}+\Lmd^{-\frac12})\log U.
\end{equation}
Suppose $U$ satisfies the negative flows \eqref{negative flow} of the extended $q$-deformed KdV hierarchy, then it is not difficult to see that the dispersionless limits of these flows coincide with the that of the flows of the Volterra hierarchy up to a certain rescaling of the times. More precisely, in the dispersionless limit, we have
\begin{equation}\label{zh-1-15}
\pfrac{}{t^{0,-p}}
=
  (-1)^p\frac{(p-1)!}{2^{2p}}\pfrac{}{\tilde{T}_p},\quad p\ge 1.
  \end{equation}
Since the flows of the Volterra hierarchy and the negative flows of the $q$-deformed KdV hierarchy are bihamiltonian systems with respect to the same bihamiltonian structure, by using the uniqueness of deformations of bihamiltonian systems of hydrodynamic type \cite{DLZ06, Liusq}, we arrive at the fact that the relations
\eqref{zh-1-15} hold true for the full flows, i.e., under the change of the unknown function \eqref{240119-2005} we have
\[
\pfrac{\tilde{\mcalL}}{t^{0,-p}}
=
  (-1)^p\frac{(p-1)!}{2^{2p}}
  \left[
   \left(
     \tilde{\mcalL}^{2p}
   \right)_+ , \tilde{\mcalL}
 \right].
\]
This fact together with the result of the last subsection confirms the Conjecture 11.1 of \cite{GFM}.

\begin{rmk}
The above-mentioned relation between the negative flows of the $q$-deformed KdV hierarchy and that of the Volterra hierarchy is given in \cite{Kup} for the first flows of these hierarchies. Due to this relation, the first negative flow of the $q$-deformed KdV hierarchy (see \eqref{t0,-1 flow}) is called the modified Volterra equation in \cite{Kup}.
\end{rmk}

\section{Conclusion}
In this paper we prove the existence and uniqueness of solution to the loop equation of a semisimple generalized Frobenius manifold with non-flat unity. The solution to the loop equation yields a quasi-Miura transformation which defines the topological deformation of the Principal Hierarchy of the generalized Frobenius manifold. We show that the topological deformation of the Principal Hierarchy of a special 1-dimensional generalized Frobenius manifold is the extended $q$-deformed KdV hierarchy, part of this result gives a proof of the Conjecture 11.1 of \cite{GFM}.  It is also conjectured in \cite{GFM} that the topological deformation of the Principal Hierarchy of a special 2-dimensional generalized Frobenius manifold contains the flows of the Ablowitz--Ladik hierarchy. We will give a proof of this conjecture in a separate publication.

As we already know, for a usual semisimple Frobenius manifold the topological deformation of the Principal Hierarchy is a bihamiltonian integrable hierarchy of KdV type, i.e., both of the flows and the bihamiltonian structure can be represented in terms of differential polynomials of the unknown functions \cite{normalform, Buryak2012, Shadrin2022, LWZ1}. This polynomiality of the integrable hierarchy is quite nontrivial, since the coefficients of $\ve^{2g}$ of the quasi-Miura transformation that relates
the Principal Hierarchy and its topological deformation are rational functions of the $x$-derivatives of the unknown functions.
Nonetheless, we conjecture that the topological deformation of the Principal Hierarchy of a semisimple generalized Frobenius manifold with non-flat unity also possesses such a property of polynomiality, and we will study this topic in subsequent publications.

\vskip 0.3truecm
\noindent\textbf{Acknowledgement.}\, This work is supported by NSFC No.\,12171268. The authors would like to thank Di Yang, Chunhui Zhou and Zhe Wang for very helpful discussions on this work.
\vskip 0.3truecm 

\noindent\textbf{Data availability statement} Data sharing not applicable to this article as no datasets were generated or analysed during the current study.
\vskip 0.3truecm

\noindent\textbf{Declarations}
\vskip 0.2truecm
\noindent\textbf{Conflict of interest} The authors have no competing interests to declare that are relevant to the content of this article.

\end{document}